\documentclass[a4paper,onecolumn,accepted=2021-03-31,11pt]{quantumarticle}
\pdfoutput=1
\usepackage[utf8]{inputenc}
\usepackage[english]{babel}
\usepackage[T1]{fontenc}
\usepackage{amsmath,amsfonts,amsthm,amssymb,mathtools}
\usepackage{bbm,bm}
\usepackage{hyperref}
\usepackage[numbers,sort&compress]{natbib}

\usepackage{tikz}
\usepackage{lipsum}

\newtheorem{theorem}{Theorem}
\newtheorem*{theorem*}{Theorem}
\newtheorem{lemma}{Lemma}
\newtheorem{proposition}{Proposition}

\newtheorem*{remark}{\textit{Remark}}
\newtheorem{definition}{Definition}

\def\beq{\begin{eqnarray}}
\def\eeq{\end{eqnarray}}
\def\blem{\begin{lemma}}
\def\elem{\end{lemma}}
\def\bprop{\begin{proposition}}
\def\eprop{\end{proposition}}
\def\bprop{\begin{proposition}}
\def\eprop{\end{proposition}}

\newcommand{\tr}{{\mathrm{tr}}}

\renewcommand{\rho}{\varrho}
\newcommand{\Hl}{\mathcal{H}}

\newcommand{\Rl}{\mathbb{R}} 
\newcommand{\Cl}{\mathbb{C}}
\newcommand{\Id}{\openone}
\newcommand{\ket}[1]{ | #1 \rangle}
\newcommand{\bra}[1]{ \langle #1 |}
\newcommand{\proj}[1]{\ket{#1}\hspace{-2.5pt}\bra{#1}}

\newcommand{\bk}[2]{\langle #1  |  #2 \rangle}

\newcommand{\vn}{\mathbf{n}} 
\newcommand{\vnt}{\mathbf{\tilde{n}}}
\renewcommand{\epsilon}{\varepsilon}

\newcommand{\norm}[1]{\left|\hspace{-1pt}\left|\hspace{1pt} #1 \hspace{1pt}\right|\hspace{-1pt}\right|}
\newcommand{\Mb}{\mathbf{M}}

\newcommand{\Mbt}{\mathbf{\tilde{M}}}

\newcommand{\td}[2]{\mathrm{d}_{\mathrm{tr}}( #1 , #2 )} 

\newcommand{\braket}[2]{\langle #1 | #2 \rangle}

\newtheorem{innercustomgeneric}{\customgenericname}
\providecommand{\customgenericname}{}
\newcommand{\newcustomtheorem}[2]{%
  \newenvironment{#1}[1]
  {%
   \renewcommand\customgenericname{#2}%
   \renewcommand\theinnercustomgeneric{##1}%
   \innercustomgeneric
  }
  {\endinnercustomgeneric}
}

\newcustomtheorem{customthm}{Theorem}
\newcustomtheorem{customlemma}{Lemma}

\begin{document}

\title{A universal scheme for robust self-testing in the prepare-and-measure scenario}

\author{Nikolai Miklin} \email{nikolai.miklin@ug.edu.pl}
\affiliation{Institute of Theoretical Physics and Astrophysics, National Quantum Information Center, Faculty of Mathematics, Physics and Informatics, University of Gdansk, 80-306 Gda\'nsk, Poland}
\affiliation{International Centre for Theory of Quantum Technologies (ICTQT), University of Gdansk, 80-308 Gda\'nsk, Poland}
\orcid{0000-0001-8046-382X}
\author{Micha{\l} Oszmaniec} \email{oszmaniec@cft.edu.pl}\affiliation{Center for Theoretical Physics, Polish Academy of Sciences, Al.~Lotnik\'o{}w 32/46, 02-668
Warszawa, Poland}
\maketitle

\begin{abstract}
We consider the problem of certification of arbitrary ensembles of pure states and projective measurements solely from the experimental statistics in the prepare-and-measure scenario assuming the upper bound on the dimension of the Hilbert space. To this aim, we propose a universal and intuitive scheme based on establishing perfect correlations between target states and suitably-chosen projective measurements. The method works in all finite dimensions and allows for robust certification of the overlaps between arbitrary preparation states and between the corresponding measurement operators. Finally, we prove that for qubits, our technique can be used to robustly self-test arbitrary configurations of pure quantum states and projective measurements. These results pave the way towards the practical application of the prepare-and-measure paradigm to certification of quantum devices.
\end{abstract}

\section{Introduction}

Quantum devices are becoming more and more complex and the possibilities of their precise control and manipulation keep increasing. Recently reported demonstration of quantum computational advantage by Google~\cite{Suprem2019} is only an intermediate milestone and  quantum technologies have a  potential real-life applications in fields such as quantum sensing~\cite{SensReview2017}, simulation of quantum systems~\cite{SimReview2014}, efficient computation~\cite{RevAlg2016} and machine learning~\cite{RevQuantML2017,RevQML2017_2}.

With the increasing complexity of quantum systems, there is a growing need for certification and verification of their performance. This task is usually realized via the combination of quantum tomography and various benchmarking schemes (see~\cite{Certification2019} for a recent review). However, these methods, despite being  powerful and universally applicable, depend on the assumptions about the inner workings of quantum systems, such as perfect measurements or  uncorrelated and independent errors.	In contrast to these approaches \emph{self-testing} is a method which  aims at proving the uniqueness of the implemented states or measurements based solely on the observed statistics and under minimal physical assumptions.  
	
The paradigm of self-testing was first introduced in the context of quantum cryptography~\cite{mayers2003self},  with the aim to obtain trust in cryptographic devices (see~\cite{vsupic2019self} for a recent review).  Initially, it was applied to correlations observed in the Bell scenario~\cite{bell1964einstein}  (see e.g.,~\cite{mayers2003self,chen2016natural,coladangelo2017all,bowles2018device}). The most know result in this area is certification of the singlet state in the case of maximal violation of the Clauser-Horne-Shimony-Holt Bell inequality~\cite{popescu1992states}. 

With the growing number of results on self-testing, more and more attention is being drawn to prepare-and-measure scenarios, that are more experimentally appealing as compared to the one of Bell (see  e.g.,~\cite{ahrens2014experimental,brask2017megahertz,aguilar2018certifying,anwer2020experimental}).
Therein, one does not need to ensure space-like separation of the measurement events by two parties; in contrast, one party, Alice, communicates some of her states to Bob, who measures them in some basis of his choice. In order to get meaningful certification results further assumptions are needed. In the most commonly studied \emph{semi-device-independent} (SDI) scenario~\cite{SDIQKD2011}, one assumes that the dimension of the quantum system used for transmitting information is bounded from above. There exist, however, alternative approaches based on other constraints like minimal overlap~\cite{brask2017megahertz}, mean energy constraint~\cite{van2017semi} or entropy constraint~\cite{chaves2015device}.

It is important to point out that without the assumption on the dimension, or an equivalent, no certification of quantum systems is possible in the prepare-and-measure scenario. The main reason is the fact that any observed distribution resulting from quantum preparations and quantum measurements can also be obtained through classical communication if the latter is unconstrained (see, e.g., Ref.~\cite{fritz2010quantum}). In this work, we chose the assumption of the bounded dimension as we find it physically motivated. For example, if photons' polarization is used as a carrier of quantum information, it is natural to assume this upper bound to be equal to two. By ensuring that no other degree of freedom of photons is correlated with the choice of her preparation, Alice can gain trust that no additional information can be leaked to Bob or an adversary.	

Let us briefly recap on what has been done so far in the area of SDI self-testing.
First, self-testing results were proven for mutually unbiased bases (MUBs) in dimension $2$, for both state ensembles and measurements~\cite{tavakoli2018self}. This was further generalised to SDI certification of mutual unbiasedness of pairs of bases in an arbitrary dimension in Ref.~\cite{farkas2019self}. Methods for self-testing of extremal qubit positive operator-valued measures (POVMs) were proposed in~\cite{mironowicz2019experimentally,tavakoli2018self2} and further extended to  symmetric-informationally complete (SIC) POVMs \cite{tavakoli2019enabling}.  Importantly, all of the above results either rely on numerical approaches, for general state preparations and POVMs, or work only for special scenarios that exhibit many symmetries. 

In this work we propose a simple analytical method allowing to certify overlaps between preparations of arbitrary pure states and arbitrary projective measurements in \emph{qudit} systems. The scheme relies on establishing perfect correlations between preparation states and outcomes of suitably-chosen projective measurements.  The method is universally applicable and robust to experimental noise. We prove that  for qubits our SDI certification method can be used to obtain a robust self-testing result for arbitrary preparations of pure qubit states and corresponding projective measurements. While for higher dimensions we do not show self-testing, our scheme allows for SDI certification of numerous inherently quantum properties. The examples include, but are not limited to:  \emph{arbitrary} overlaps between any number of measurement bases,  MUB conditions, information-completeness of measurements, and SIC relations among measurement effects.  

We believe that our findings greatly extend the applicability of the  paradigm of self-testing in the SDI setting. Our approach should be of interest to the researchers working on certification of near-term quantum computers \cite{Certification2019} since our scheme can be used to characterize state-preparation and measurement errors, which is one of the major problems in the certification of quantum devices \cite{GateSetTomography2017}. We also expect possible cryptographic applications as our setup is very similar to the one of textbook quantum key distribution schemes \cite{bb84,bennett1992quantum}. The perfect correlations can be utilized for generation of the secret key, while the rest can be used to estimate the security. Thus, our methods can be directly applied for certification of quantum devices implementing protocols such as BB84~\cite{bb84}, which is normally achieved by introducing additional preparation states or measurement bases~\cite{ErikQKD}.

Before we proceed, let us first fix some of the notations used in this paper. Let $X$ be a linear operator acting on a finite-dimensional Hilbert space $\Hl$. Throughout the paper we use $\norm{X}$, $\norm{X}_F$ to denote operator norm and Frobenius norm of $X$. We will also use $|\vn|$ to denote the Euclidean norm of $\vn\in\Rl^3$, and $[n]$ to denote an $n$-element set $\{1,2,\dots,n\}$.

\section{Description of the scenario}
	
We consider the prepare-and-measure scenario in which in each run of the experiment Alice prepares a quantum $d$-level system in one of the states from a finite set of preparations $\rho_a^x$, for which we use two indexes $x\in [n]$ and $a\in [d]$. The purpose of the double index will become clear later. Subsequently, Bob performs a measurement on this state with a finite choice of measurement settings $y\in[n]$ having the possible outcomes $b\in[d]$. We assume that the parties do not communicate in any other way and do not have access to any entangled states or shared randomness~\cite{deVincente2017}. This implies the standard i.i.d.~assumption that in all the experiment runs, the same choices of preparations lead to the identical and uncorrelated quantum states prepared by Alice's device. The equivalent assumption is placed on Bob's measuring device. The role of these assumptions is discussed later in the text. This also implies that the observed statistics $p(b|a,x,y)$ are given by the Born rule i.e.,~$p(b|a,x,y)= \tr(\rho^x_a M^y_b)$, where $\Mb^y = (M^y_1,M^y_2,\dots,M^y_d)$ is a quantum measurement (POVM) performed by Bob upon the choice of the setting $y$. The goal of SDI certification is then to identify the preparation states and the measurements, or their properties, based solely on the observed statistics $p(b|a,x,y)$ \emph{assuming} the upper bound on the dimension $d$ and the validity of Born's rule. We say that certain states $\rho^x_a$ and measurements $\Mb^y$  can be \emph{self-tested} if the observed statistics specify these objects uniquely up to a unitary transformation and, perhaps, a global transposition.   

\begin{figure}[t!] \centering
	\includegraphics[width=.5\textwidth]{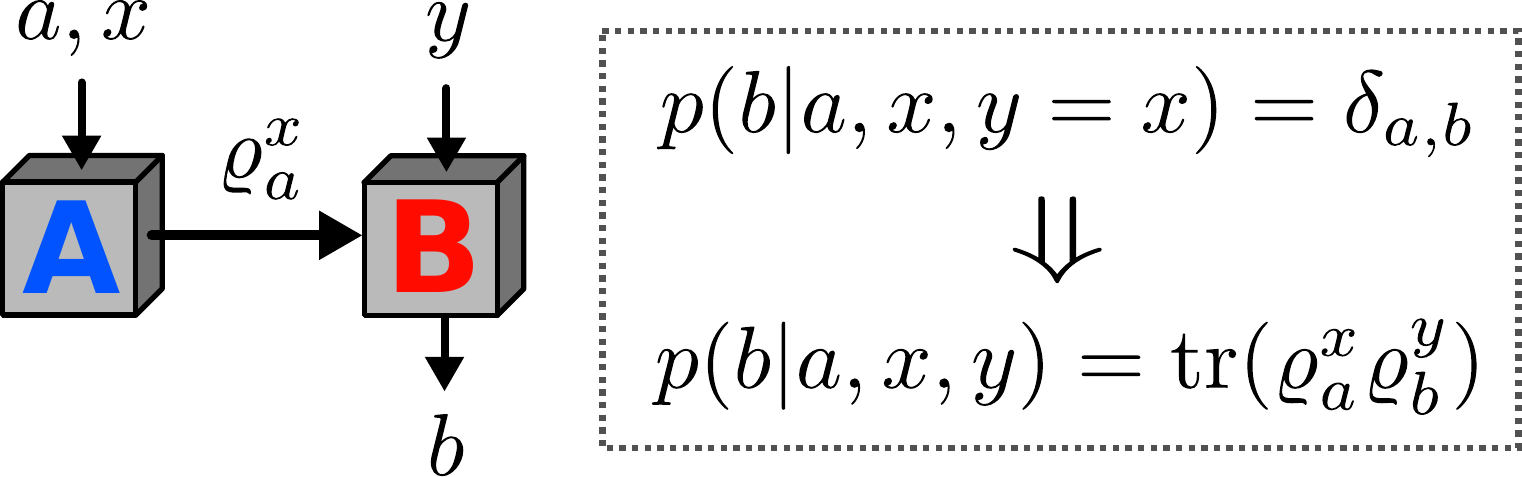}
	\caption{The idea of the certification scheme for the overlaps between states. Alice chooses her inputs $a,x$ and Bob his input $y$. Alice sends states $\rho^x_a$ to Bob, who produces his outcome $b$. After establishing perfect correlations between preparation states and measurements for $y=x$, the rest of the statistics can be used to compute the overlaps between the preparation states. This can be used to directly certify presence of coherences between Alice's input states and between effects of measurements implemented by Bob.}
	\label{fig:scheme}
\end{figure}
	
\section{Certification of overlaps} 	We start with a presentation of our scheme for certifying pairwise overlaps between \emph{pure qudit} preparation states and between the corresponding \emph{projective} measurements. By ``corresponding" we mean that in our scheme we set these states and measurements to be ``equal" in a sense that $\varrho^x_a = (\varrho^x_a)^2 = M^x_a$ for all $a\in[d]$ and $x\in[n]$. This is the reason why we chose to enumerate the preparation states with two indexes. The choice of $d$ and $n$ is arbitrary as long as $d>1$ and $n\geq 1$. In what follows we will refer to these objects as target pure states and target measurements respectively. Their experimental counterparts we denote as $\tilde{\rho}_a^x$, and $\Mbt^y$ respectively. By ``experimental" we mean any states and measurements defined over a Hilbert space of dimension at most $d$ and that reproduce the \emph{observed statistics} i.e.,~$\tilde{p}(b|a,x,y)=\tr(\tilde{\rho}^x_a \tilde{M}^y_b)$. Clearly, we \emph{do not} assume that the experimental states and measurement have to be ``equal". Neither do we assume any properties of experimental states or measurements like the purity of the projective character.

The idea of our certification scheme is very intuitive, yet powerful (see Fig.~\ref{fig:scheme}). Assume that Alice and Bob prepared their devices in a way that $\tilde{p}(b|a,x,y)=1$, whenever $y=x$ and $b=a$. In other words, outcomes of Bob's measurement are \emph{perfectly correlated} with the preparations of Alice (whenever $x=y$). Since we assumed the upper bound of $d$ on the Hilbert space's dimension in which experimental states and measurements are defined, and in this space, we have a set of $d$ quantum states and a POVM with $d$ effects which produce perfect correlations, we can conclude that, first, the Hilbert space's dimension must be exactly $d$, and second that $\tilde{\varrho}^x_a = \tilde{M}^x_a$, for all $a\in[d]$ and $x\in[n]$. Clearly, after these perfect correlations are established, the ``cross-terms'',  can be used to compute the overlaps between the preparation states and between measurement operators: $\tilde{p}(b|a,x,y)=\tr(\tilde{\rho}^x_a \tilde{M}^y_b)=\tr(\tilde{\rho}^x_a \tilde{\rho}^y_b)=\tr(\tilde{M}^x_a \tilde{M}^y_b)$. Therefore, if the experimental statistics $\tilde{p}(b|a,x,y)$ match the target statistics $p(b|a,x,y)$, 	we can certify that overlaps between experimental states match those of the target states. The same holds for the corresponding measurement operators.

Our method can be also applied when experimental statistics do not match exactly the target ones. 

\begin{theorem}[Robust SDI certification of overlaps]\label{th:overlaps}
Consider  pure target qudit preparation states $\rho^x_a$ and target projective measurements $\Mb^y$,  where $a\in [d]$ and $x,y\in [n]$.  Assume that $\rho^x_a = M^x_a$ for all $a,x$ and furthermore that experimental states $\tilde{\rho}^x_a$ and measurements $\Mbt^y$ act on Hilbert space of dimension at most $d$ and generate statistics $\tilde{p}(b|a,x,y)=\tr(\tilde{\rho}^x_a\tilde{M}^y_b)$ such that $|\tilde{p}(b|a,x,y)- \tr(\rho^x_a M^y_b)|\leq \epsilon$, for all  $a,b,x,y$.   Then, input states $\tilde{\rho}^x_a$ are almost pure and measurements  $\Mbt^y$ are almost projective in the sense that 
\beq
\label{eq:purity}
\text{for all $x$\ \     } \sum_{a=1}^{d}\norm{\tilde{\rho}^x_a} \geq d(1-2\epsilon)\ ,\\
\text{for all $y$\ \    }\sum_{b=1}^d\norm{\tilde{M}^y_b} \geq d(1-\epsilon)\  \label{eq:projectivness} .
\eeq
Moreover, for all $x\neq x'$, $a\neq a'$, $y\neq y'$, and $b\neq b'$,  we have
\beq \label{eq:overlaps}
&&|\tr(\tilde{\rho}^x_a\tilde{\rho}^{x'}_{a'})-\tr(\rho^x_a\rho^{x'}_{a'})|\leq \epsilon+\sqrt{2\epsilon+d^2\epsilon^2},\\
&&|\tr(\tilde{M}^y_b\tilde{M}^{y'}_{b'})-\tr(M^y_bM^{y'}_{b'})|\leq \epsilon+(1+d\epsilon)\sqrt{2\epsilon+d^2\epsilon^2}\ . \nonumber
\eeq
\end{theorem}
The parameter $\epsilon$ accounts for all possible experimental errors, including the finite statistics' error that can be estimated using standard tools of statistical analysis.
\begin{proof}
We start with a straightforward proof of Eq.\eqref{eq:projectivness}. Since $\norm{\tilde{M}^x_a}\geq \tr(\tilde{M}^x_a\rho)$ for any state $\rho$, the relation $\tr(\tilde{\rho}^x_a\tilde{M}^x_a)\geq 1-\epsilon$, valid for all $a,x$, implies $\norm{\tilde{M}^x_a}\geq 1-\epsilon$, for all  $a,x$, and hence for all $x$ we have $\sum_{a=1}^{d}\norm{\tilde{M}^x_a} \geq d-d\epsilon$.

To prove the bounds on the norms of experimental states in Eq. \eqref{eq:purity}, we fix a setting $x$ and use a decomposition: 
\begin{equation}\label{eq:partSPEC}
\tilde{M}^x_a = \lambda_a\proj{\phi_a}+\text{Rest}_a\ ,
\end{equation}
where $\lambda_a$ is the largest eigenvalue of $\tilde{M}^x_a$, $\proj{\phi_a}$ is the corresponding eigenvector and,  $\text{Rest}_a$ is a positive-semi-definite  operator satisfying  $\bra{\phi_a}\text{Rest}_a\ket{\phi_a} = 0$ (to keep the derivations legible, we omit the index $x$). Using the fact that observed statistics are $\epsilon$ close to the target ones we get:
\beq
1-\epsilon\leq \tr(\tilde{\rho}^x_a\tilde{M}^x_a) = \lambda_a\bra{\phi_a}\tilde{\rho}^x_a\ket{\phi_a}+\tr(\tilde{\rho}^x_a\text{Rest}_a)\leq \norm{\tilde{\rho}^x_a}+\tr(\text{Rest}_a)\ . 
\eeq
By taking the sum over $a$ we obtain $\sum_{a=1}^d\norm{\tilde{\rho}^x_a}\geq d(1-\epsilon)-\sum_{a=1}^d\tr(\text{Rest}_a)$. We can now use the identity $d = \sum_{a=1}^d\tr(\tilde{M}^x_a)$, which follows form that fact that operators $M^x_a$ form a POVM in $\Cl^d$. This equality allows us to give a bound: 
\begin{equation}
  d=   \sum_{a=1}^d\lambda_a+\sum_{a=1}^d\tr(\text{Rest}_a)\leq d(1-\epsilon)+\sum_{a=1}^d\tr(\text{Rest}_a) \ ,
\end{equation}
which is equivalent to $\sum_{a=1}^d\tr(\text{Rest}_a)\leq d\epsilon$. Inserting this to $\sum_{a=1}^d\norm{\tilde{\rho}^x_a}\geq d(1-\epsilon)-\sum_{a=1}^d\tr(\text{Rest}_a)$ we obtain Eq.~\eqref{eq:purity}.

We now proceed to the proof of Eqs.~\eqref{eq:overlaps}. We start with the one for the overlaps between preparation states. The proof is given by the following sequence of inequalities which hold for every $a\neq a'$, $x\neq x'$:
\beq
\label{eq:app_proof_overlaps}
&&|\tr(\tilde{\rho}^x_a\tilde{\rho}^{x'}_{a'}) - \tr(\rho^x_a\rho^{x'}_{a'})|\leq |\tr(\tilde{\rho}^x_a\tilde{\rho}^{x'}_{a'}) - \tr(\tilde{M}^x_a\tilde{\rho}^{x'}_{a'})|+ |\tr(\tilde{M}^x_a\tilde{\rho}^{x'}_{a'}) - \tr(M^x_a\rho^{x'}_{a'})|\nonumber\\
&&\leq \norm{\tilde{\rho}^x_a-\tilde{M}^x_a}+\epsilon \leq \norm{\tilde{\rho}^x_a-\tilde{M}^x_a}_F+\epsilon= \sqrt{\tr((\tilde{\rho}^x_a)^2)+\tr((\tilde{M}^x_a)^2)-2\tr(\tilde{\rho}^x_a\tilde{M}^x_a)}+\epsilon\nonumber\\
&&\leq \sqrt{1+(1+d^2\epsilon^2)-2(1-\epsilon)} = \epsilon+\sqrt{2\epsilon+d^2\epsilon^2}.\nonumber
\eeq
All of the above inequalities are pretty straightforward apart from $\tr((\tilde{M}^x_a)^2)\leq 1+d^2\epsilon^2$ that we prove below. For this we again use the partial spectral decomposition form Eq.~\eqref{eq:partSPEC} (without writing the superscript $x$ as above) which implies:
\beq
\tr((\tilde{M}^x_a)^2) = \lambda_a^2+\tr(\text{Rest}_a^2)\leq 1+(\tr(\text{Rest}_a))^2\leq 1+d^2\epsilon^2.
\eeq
The proof for the overlaps of POVM effects is very similar to the one for states with the only difference being the following inequality:
\beq
|\tr(\tilde{M}^x_a\tilde{M}^{x'}_{a'}) - \tr(\tilde{M}^x_a\tilde{\rho}^{x'}_{a'})| \leq (1+d\epsilon)\norm{\tilde{\rho}^x_a-\tilde{M}^x_a},
\eeq
that is used in the second step in the proof in Eq.~(\ref{eq:app_proof_overlaps}). One can easily verify the validity of this inequality by writing once more the decomposition $\tilde{M}^{x'}_{a'} = \lambda_a\proj{\phi_{a'}}+\text{Rest}_{a'}$ and remembering that $\tr(\text{Rest}_a)\leq d\epsilon$.
\end{proof}

The above result states that if the statistics observed  in our certification scheme vary just a little bit form the target ones, the overlaps of the experimental states are also close to the overlaps between target states (and analogously for measurements). To our best knowledge analogous results have been previously known only for very special symmetric target states and measurements forming MUBs \cite{tavakoli2018self,farkas2019self}. 

In Appendix~\ref{app:rob} (Lemma~\ref{lemma:overlaps}) we improve the above bounds for the case of qubit systems.
Moreover, in Appendix~\ref{app:remark} we prove, by giving explicit examples, that the bounds in Eq.~\eqref{eq:overlaps} are tight in the first orders in $\sqrt{\epsilon}$ and $d$. 

\section{Self-testing of qubits} We now show that certification of overlaps allows to prove robust self-testing result for \emph{arbitrary} pure qubit preparations and projective measurements appearing in our certification scheme. 

\begin{theorem}[Robust self-testing of qubit systems]
\label{th:rob}
Consider target pure qubit states $\rho^x_a$ and projective measurements  $\Mb^y$, where $a=1,2$, $x,y\in  [n]$, and  $\rho^x_a = M^x_a$ for all $a,x$. Assume that experimental qubit states $\tilde{\rho}^x_a$ and measurements $\Mbt^y$ generate statistics $\tilde{p}(b|a,x,y)=\mathrm{tr}(\tilde{\rho}^x_a\tilde{M}^y_b)$ such that $|\tilde{p}(b|a,x,y)- \tr(\rho^x_a M^y_b)|\leq \epsilon$, for all $a,b=1,2$ and $x,y\in[n]$. Then, there exist $\epsilon_0$ such that for $\epsilon\leq \epsilon_0$  there exist a qubit unitary $U$ such that 
\beq
\label{eq:th_rob}
\frac{1}{2n} \sum_
{a,x}\tr(U(\tilde{\rho}^x_a)^{(T)} U^\dagger\rho^x_a) \geq 1 - f(\epsilon)\ , \nonumber\\
\frac{1}{2n} \sum_
{b,y}  \tr(U(\tilde{M}^y_b)^{(T)} U^\dagger M^y_b) \geq 1 -g(\epsilon) \ ,
\eeq
where $(\cdot)^{(T)}$ is a potential transposition with respect to a fixed basis in $\Cl^2$ that may have to be applied to all experimental states and measurements at the same time. Moreover, functions $f,g:[0,\epsilon_0) \rightarrow \Rl_+$ depend solely on the target states and measurements and, for small $\epsilon$, have the asymptotics $f(\epsilon)\propto \epsilon,\ g(\epsilon)\propto \epsilon$. 
\end{theorem}

\noindent In the above we used the fidelity $F(\rho,\sigma) = \tr(\rho \sigma)$ to indicate the closeness between \emph{rotated} experimental states and target pure states (and analogously for measurements), following existing literature. The case $\epsilon =0$ corresponds to ideal reconstruction of target states and measurements after applying a suitable isometry.  
We remark that we allow only unitary operations (and possible transposition), as opposed to general channels \cite{tavakoli2018self,tavakoli2018self2}, to be applied to the experimental states in order to approximate the target states \emph{as well as possible}.

In Appendix~\ref{app:rob} we give a formal version of the above result and its proof. Moreover, we present there robustness bounds expressed in terms of the trace distance and its analogue for measurements~\cite{NavasquesMES}.  We remark that the functions $f,g$ become unbounded once Bloch vectors of target qubit states become singular, i.e.,~once the target states are close to being aligned in a space of smaller dimension.

\begin{remark}
Certification of overlaps between pure states in general does not allow for their  self-testing in higher-dimensional systems.  This is e.g.,~due to the existence of unitary inequivalent sets of SIC-POVMs for $d=3$ \cite{SICpaper2004} and MUBs for $d=4$~\cite{brierley2009mutually} (even if we allow for complex conjugation).  
\end{remark}

\begin{proof}[Sketch of the proof] We give a full proof for the ideal case ($\epsilon=0$) below. From Theorem~\ref{th:overlaps} it follows that for all $x,x',a,a'$ we have $\tilde{\rho^x_a} = \tilde{M}^x_a$ and $\tr(\tilde{\rho}^x_a\tilde{\rho}^{x'}_{a'}) = \tr(\rho^x_a\rho^{x'}_{a'})$. Using the Bloch representation, $\rho^x_a = \frac{1}{2}(\Id + \mathbf{n}^x_a\cdot\bm{\sigma})$, we can conclude that also $\mathbf{\tilde{n}}^x_a\cdot\mathbf{\tilde{n}}^{x'}_{a'} =  \mathbf{n}^x_a\cdot\mathbf{n}^{x'}_{a'}$, where $\vn^x_a$ and $\tilde{\vn}^x_a$ are Bloch vectors of  $\rho^x_a$ and $\tilde{\rho}^x_a$ respectively. 

Assume now that the vectors $\vn^1_1,\vn^2_1,\vn^3_1$ are linear independent. Let $O$ be the linear transformation defined by  $O\mathbf{n}^x_1 = \mathbf{\tilde{n}}^x_1$, $x=1,2,3$, and let $L$ be the matrix whose rows are the vectors $\vn^x_1$, $x=1,2,3$. Then, we have $LO^TOL^T = LL^T$, and consequently, since $L$ is invertible by the construction, $O^TO = \Id_3$, i.e.,~$O$ is an orthogonal transformation in $\Rl^3$. It is well-known~\cite{NielsenBook} that if $\det(O) = 1$, there exist a unitary matrix $U$  such that $\tr(U(\tilde{\rho}^x_a)U^\dagger\rho^x_a)=1$ for $x=1,2,3$ and $a=1$. By our assumption all remaining states $\rho^x_a$ can be decomposed in the basis $\{\Id,\rho^1_1,\rho^2_1,\rho^3_1\}$, with the coefficients depending solely on the overlaps $\tr(\rho^x_a\rho^{x'}_{a'})$. The, the same $U$ that maps $\tilde{\rho}^x_1$ to $\rho^x_1$, for $x=1,2,3$, also connects the remaining pairs of states. Finally, if $\det(O) = -1$, the transformation $O$ corresponds to application of the transposition in the standard basis of $\Cl^2$ followed by application of a unitary operation $U$, determined by $O$~\cite{BergmannWigner1964}. 

For the general case $\epsilon>0$ we consider Cholesky factorisations of Gram matrices, denoted by $\Gamma$ and $\tilde{\Gamma}$, of Bloch vectors $\vn^1_1,\vn^2_1,\vn^3_1$ and their experimental counterparts respectively. Theorem~\ref{th:overlaps} is then used to bound $\norm{\Gamma-\tilde{\Gamma}}_F$ and utilize results of Ref.~\cite{sun1991perturbation} to gauge how much the Cholesky decompositions of $\Gamma$ and $\tilde{\Gamma}$ differ in the Frobenius norm. The latter can be connected to the average fidelity  between the selected target and experimental states. The robustness for the remaining states follows from the fact that they can be decomposed (as operators) using the three initially chosen target states and the identity. 

If vectors from the set $\lbrace{ \vn^x_a\rbrace}$ span two dimensional space then the same arguments for both $\epsilon=0$ and $\epsilon>0$ can be repeated for the considered subspace. Importantly, in this case additional transposition is not necessary.
\end{proof}

\begin{figure}[t!] \centering
    \includegraphics[width=.95\textwidth]{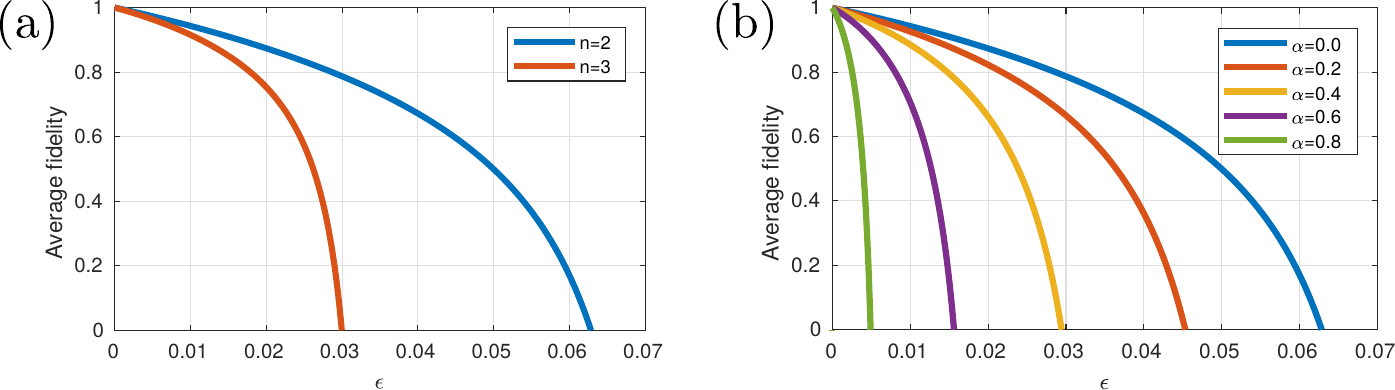}\caption{Lower bounds for the average fidelity between the experimental and target states. Part (a) presents results for $n=2,3$ qubit MUBs. Part (b) presents results for two qubit bases for different degree of bias $\alpha$. Value $\alpha=0$ corresponds to two MUBs while $\alpha=1$ gives two identical bases.}
	\label{fig:ex}
\end{figure}

\section{Examples} We now apply the quantitative formulation of Theorem \ref{th:rob} to lower-bound average fidelities for different configurations of target states as a function of the allowed error $\epsilon$. In Figure~\ref{fig:ex} we present results for  eigenstates of $n=2$ and $n=3$ Pauli matrices (i.e.,~states forming $n=2$ and $n=3$  qubit MUBs), and states belonging to two biased projective measurements satisfying $\tr{(\rho^1_a\rho^2_a)} = \frac{1+\alpha}{2}$, $a=1,2$ where we take $\alpha\in[0,1]$. 

For $n=2$ MUBs we compare our results with Ref.~\cite{tavakoli2018self} that focuses on self-testing of qubit MUBs. The results of Ref.~\cite{tavakoli2018self} give the upper bound on average fidelity equal $0.75$ for the deviation of $\simeq 0.1$ in the figure of merit. In our scheme this happens for $\epsilon\simeq 0.033$ as shown on Fig.~\ref{fig:ex}. To test the versatility of our scheme, we also applied it to $n=3$ MUBs, trine and tetrahedral configurations of qubit states. Quantitative results concerning these examples are listed in Table~\ref{tab:examples}, while detailed derivations are
given in Appendix~\ref{app:ex}. For trine and tetrahedral configurations the robustness is obtained via the so-called Procrustes problem, described in Appendix~\ref{app:procrust}, based on Ref.~\cite{arias2020perturbation}. For cases other than two MUBs we cannot make any comparison with the existing literature  since, to our best knowledge, these case have not been studied previously.

\begin{table}\centering
\begin{tabular}{|c|c|c|}
\hline
Configuration      & $\epsilon_0$ & $C$  \\ \hline
2 MUBs      &      $\approx 0.062$        &   $\frac{7}{2}+\sqrt{2}$     \\ \hline
3 MUBs      &      $\approx 0.030$        &  $6$       \\ \hline
2 biased bases   & $\lesssim \frac{4-3\alpha-\sqrt{7-6\alpha}}{18}$   &  $ 2+\left(1 + \frac{\sqrt{1+\alpha}}{\sqrt{2}(1-\alpha)}\right)^2$   \\ \hline
Trine       &   $\approx 0.058$    &  $\frac{19}{3}$ \\ \hline
Tetrahedron &    $\approx 0.037$       & $10$ \\ \hline
\end{tabular}

\caption{\label{tab:examples} Results of quantitative variant of Theorem~\ref{th:rob} applied to different configurations of target quantum states. The threshold $\epsilon_0$ sets the maximal noise level which is tolerated by our scheme. The constant $C$ is defined via the relation $f(\epsilon) \stackrel{\epsilon\rightarrow 0}{\approx} C\epsilon$, where $1-f$ is the lower bound on the average fidelity from Eq. \eqref{eq:th_rob}. 
}
\end{table} 

 \section{Shared randomness} Throughout the article we assumed that preparation and measurement devices are uncorrelated, which is often true in practice. However, one may also consider a situation in which Alice and Bob share a random variable $\lambda$, that they can use to coordinate their actions in a given round of the experiment. The most general statistics that can be generated in such a scenario can be expressed as  $p(b|a,x,y)= \int d\lambda p(\lambda)   \tr(\rho^{x}_{a}(\lambda) M^{y}_{b}(\lambda))$, where $p(\lambda)$ denotes the probability distribution of $\lambda$.

Interestingly, the presence of shared randomness makes our main results, as stated by Theorems~\ref{th:overlaps} and~\ref{th:rob}, inapplicable. The easiest way to see it is to consider the following simple example of $n=d=2$, one bit of shared randomness $\lambda \in \{1,2\}$, and $\tilde{\rho}^{1,2}_a = \proj{a}$, $\tilde{M}^{1,2}_b = \proj{b}$. This clearly satisfies the requirement on $\tilde{p}(a|a,x,x)=1$. Now, Alice's and Bob's devices can decide to ``flip" their preparations and measurement operators whenever $x=2$, $y=2$ and $\lambda = 2$ (note that $\lambda$ can be distributed in an arbitrary way). This procedure does not affect the correlations $\tilde{p}(a|a,x,y=x)$, but it can be used to set $\tilde{p}(a|a,1,2)$ to an arbitrary value. 

Additionally, we believe that in the presence of shared randomness, the usual notion of self-testing in the SDI setting has to be reconsidered. Specifically, for \emph{any} quantum realisation giving the statistics $p(b|a,x,y)= \tr(\rho^x_a M^y_b)$, one can always consider a strategy in which  with probability $p(\lambda)$ Alice prepares  states $\rho^x_a (\lambda) =U_\lambda \rho^x_a U^\dagger_\lambda$ and Bob implements measurements  $M^{y}_{b}(\lambda)= U_\lambda M^{y}_{b} U^\dagger_\lambda$, where $U_\lambda$ is some arbitrary unitary transformation depending on $\lambda$. Clearly, such strategy reproduces the original statistics, and makes it impossible to find a single unitary that connects the target and experimental states (or measurements).

A very similar argument as above can be applied to situations of non-identical preparations of states and non-identical measurements across different runs of the experiment. Indeed, the value of a given run of the experiment is accessible to both Alice and Bob and can be considered a shared random variable $\lambda$ that controls both Alice's and Bob's devices. Once again, the existence of a single unitary that connects the target and experimental states can never be guaranteed.

While, we do not consider the assumption on no shared randomness strong, below we propose a modification of our SDI certification scheme that allows to certify overlaps between arbitrary pure states even in the presence of shared randomness. The idea is to introduce, for every pair of non-orthogonal states, a suitable \emph{intermediate state}~\cite{Gisin2003}, that enforces fixed overlaps between experimental states in every round of the experiment. 

Recall that in our certification scheme we consider pure target qudit states $\rho^x_a$ and target projective measurements $\Mb^y$,  where $a\in [d]$ and $x,y\in [n]$. We extend this scheme by introducing an additional intermediate target state $\rho_z$ for every two pairs of Alice's input  $(x,a)$ , $(x',a')$, where $a\neq a'$ and $x\neq x'$. The state $\rho_z$ is chosen as the unique state satisfying: 
\begin{equation}\label{eq:intermediateSTATE}
    \tr\left(\rho_z (\rho^x_a +\rho^{x'}_{a'}) \right)=1+\sqrt{\tr(\rho^x_a \rho^{x'}_{a'})} \ .
\end{equation}
Let $\tilde{p}(b|z,y)$ denote the experimental statistics, corresponding to inputs $z,y$ and outcome $b$.  In the ideal scenario, in which experimental statistics satisfy assumptions of Theorem~\ref{th:overlaps} with $\epsilon=0$, assume additionally that $\tilde{p}(a|z,x)+\tilde{p}(a'|z,x')=1+\sqrt{\tilde{p}(a'|x,a,x')}\ .$
Below we prove that under the above assumptions for all $\lambda$ we have $\tr(\tilde{\rho}^x_a(\lambda) \tilde{\rho}^{x'}_{a'} (\lambda) )=\tr(\rho^x_a \rho^{x'}_{a'})$. Since in the ideal case overlaps between preparation states  \emph{do not depend} on $\lambda$ and match the target value, our modified certification scheme works also in the presence of shared randomness. In the noisy scenario one can generally expect to be able to upper-bound the probability (over the choice of $\lambda$) of the deviations from perfect correlations between states and measurements, as a function of the error $\epsilon$.  We postpone the discussion of this interesting point to the future work.

\begin{proof}[Proof of soundness of certification in the presence of shared randomness]
Consider a general physical realisation of experimental statistics $\tilde{p}(b|x,a,y)$, $\tilde{p}(b|z,y)$ via classically correlated preparations and measurements on $\Cl^d$:
\begin{align}
   \tilde{p}(b|x,a,y)&= \int d\lambda p(\lambda) \tr(\tilde{\rho}^x_a (\lambda) \tilde{M}^y_b (\lambda) )\ , \\
    \tilde{p}(b|z,y)&= \int d\lambda p(\lambda) \tr(\tilde{\rho}_z (\lambda) \tilde{M}^y_b (\lambda) )\ ,
\end{align}
where $x,y\in[n]$, $a,b\in[d]$ and the variable (input) $z$ labels elements in the set of unordered pairs $\lbrace{(x,a),(x',a')\rbrace}$, where $x\neq x'$ and $a\neq a'$. Assume now that the above statistics match \emph{exactly} the ones required by our scheme and are compatible with Eq.~\eqref{eq:intermediateSTATE} in the sense that:
\beq
\tilde{p}(b|x,a,y) &=& \tr(\rho^x_a \rho^y_b)\ \label{eq:basic},\\ 
    \tilde{p}(a|z,x)+\tilde{p}(a'|z,x')&=&1+\sqrt{\tilde{p}(a'|x,a,x')}\ \label{eq:extra} .
\eeq
In what follows we prove that if the above constraints are satisfied, then for all\footnote{Technically speaking this equality holds for \emph{almost all} $\lambda$.}:
\beq
\tilde{\rho}^x_a(\lambda) &=&\tilde{M}^x_a (\lambda) \ ,\label{eq:keyEQUATION}   \\
 \tr(\rho^x_a \tilde{\rho}^{x'}_{a'})&=&\tr(\rho^x_a(\lambda) \tilde{\rho}^{x'}_{a'}(\lambda) )  \ ,\label{eq:keyEQUATION2}
\eeq
where $x,x'\in[n]$, $a,a'\in[d]$. The above equation means that the overlaps between preparation sates (and measurements) do not depend on on the value of the shared random variable $\lambda$.

The proof of Eq.~\eqref{eq:keyEQUATION} is straightforward. Namely, form Eq.~\eqref{eq:basic} if follows that:
\begin{equation}
    \int d\lambda p(\lambda) \tr\left(\tilde{\rho}^x_a(\lambda) \tilde{M}^x_a (\lambda)\right) =1\ .
    \end{equation}
Since $\tr(\tilde{\rho}^x_a(\lambda) \tilde{M}^x_a (\lambda))\leq 1$, we get that for all $\lambda$ we have $\tr(\tilde{\rho}^x_a(\lambda) \tilde{M}^x_a (\lambda))=1$. Since this reasoning can be repeated for all $a\in[d]$ (for the fixed value of $x$) and since operators $\tilde{M}^x_a$ form a POVM on $\Cl^d$ , we finally get Eq.~\eqref{eq:keyEQUATION}.

The proof of Eq.~\eqref{eq:keyEQUATION2} is more involved and relies on both Eqs.~\eqref{eq:basic} and~\eqref{eq:extra}. Specifically, from Eq.~\eqref{eq:basic} and the already established identity $\tilde{\rho}^x_a = \tilde{M}^x_a$,  it follows that Eq.~\eqref{eq:extra} is equivalent to:
\begin{equation}
\int d\lambda p(\lambda) \tr\left(\tilde{\rho}_z (\lambda) [\tilde{\rho}^x_a(\lambda) +\tilde{\rho}^{x'}_{a'}(\lambda)] \right) =1 +\sqrt{\tr(\rho^x_a \rho^{x'}_{a'})}\ ,
\end{equation}
where moreover:
\begin{equation}\label{eq:Average}
    \int d\lambda p(\lambda) \tr\left(\tilde{\rho}^x_a(\lambda) \tilde{\rho}^{x'}_{a'}(\lambda)\right)  = \tr(\rho^x_a \rho^{x'}_{a'}) \ . 
\end{equation}
Using Bloch representation of qubit states (states $\tilde{\rho}^x_a(\lambda),\tilde{\rho}^{x'}_{a'}(\lambda)$ are pure and hence span a two dimensional subspace of $\Cl^d$) it is straightforward to obtain the bound: 
\begin{equation}
    \tr\left(\tilde{\rho}_z (\lambda) [\tilde{\rho}^x_a(\lambda) +\tilde{\rho}^{x'}_{a'}(\lambda)] \right) \leq 1+\sqrt{\tr(\tilde{\rho}^x_a(\lambda)\tilde{\rho}^{x'}_{a'}(\lambda))}\ .  
\end{equation}
After setting $g(\lambda)=\sqrt{\tr(\tilde{\rho}^x_a(\lambda)\tilde{\rho}^{x'}_{a'}(\lambda))}$ and using Eq.~\eqref{eq:Average} we  obtain: 
\begin{equation}
    \int d \lambda p(\lambda) g(\lambda) \geq \sqrt{\int d\lambda p(\lambda) g(\lambda)^2}\ .
\end{equation}
Using Chauchy-Schwartz inequality for the left hand side of the above inequality we finally get: 
\begin{equation}
    \int d \lambda p(\lambda) g(\lambda) = \sqrt{\int d\lambda p(\lambda) g(\lambda)^2}\ ,
\end{equation}
which is equivalent to saying that the variance of the random variable $g(\lambda)$ vanishes. Therefore, $g(\lambda)=\alpha$, where $\alpha$ is some numerical constant. We conclude the proof by noting that from Eq.~\eqref{eq:Average} it follows that $\alpha^2 = \tr(\rho^x_a \rho^{x'}_{a'})$ which implies Eq.~\eqref{eq:keyEQUATION2}.
\end{proof}

\section{Extension of the scheme to SDI characterization of rank-1 non-projective measurements}
\label{app:povms}
In our work we focus on self-testing of arbitrary configurations of preparation states and projective measurements. In this section we discuss possible extension of our results to certification of properties of general rank-1 measurements in arbitrary dimension $d$. 
To that end, we introduce additional states coming from suitably-chosen orthonormal bases per every effect of the measurement we want to certify. With the help of these extra states (and corresponding projective measurements) we can e.g., certify the following properties of POVMs in arbitrary dimension $d$:  (i) information completeness, (ii) extremality of rank-1 POVMs,  and (iii) SIC-property.

The main idea of the extension is the following. Suppose we would like to self-test an $m$-outcome POVM $\mathbf{N} = (N_1,\dots,N_m)$ acting on a Hilbert space of dimension $d<m$ with each effect $N_b$ being of rank $1$ i.e., $N_b=\alpha_b \sigma_b$, for some pure states $\sigma_b$ and positive scalars $\alpha_b$. We  consider $m$ sets of preparation states $\{\rho^x_a\}$, $a\in [d]$, ${x\in[m]}$ and the corresponding projective measurements $\mathbf{M}^y$, $y\in [m]$. Assume now that in the idealized case we observe statistics satisfying: 
\beq
\label{eq:app_rank1povm}
&&\tr(\rho^b_a N_b) = 0,\quad \forall b\in [m],\, a\neq 1,\nonumber \\
&&\tr(\rho^x_a M^x_b) = \delta_{a,b},\quad \forall x\in [m]\ .
\eeq
From the second set of conditions we conclude that $\{\rho^x_a\}_a$ forms, as before, a basis of pure states for each $x$. From the first condition in Eq.~\eqref{eq:app_rank1povm} we get that $\mathbf{N}$ has to be a rank-1 POVM and, moreover, that $\rho^b_1 =\sigma_b$ ($N_b$ is orthogonal to all of the states except the one with $a=1$). Consequently, we have $ \tr(\rho^b_1 N_b)=\alpha_b$ for all $b\in[m]$.  This fact together with the measured overlaps $\tr(\rho^b_1\rho^{b'}_1) =\tr(\rho^b_1 M^{b'}_1) $ gives the information about the overlaps $\tr(N_b N_{b'})$ between effects of the POVM $\mathbf{N}$.

The knowledge of all overlaps $\tr(N_b N_{b'})$ (including $b'=b$) and the fact that $\mathbf{N}$ is a rank-1 POVM allow to certify many interesting properties of $\mathbf{N}$. First, as long as $m\leq d^2$ we can certify extremality of $\mathbf{N}$. This follows from the fact that rank-1 POVMs are extremal if and only if their effects are linearly independent (see e.g.,~Ref.~\cite{D_Ariano_2005}). Linear independence can be directly inferred from the experimentally accessible moment matrix $\Gamma_{b,b'} = \tr(N_b N_{b'})$. Specifically, $\Gamma$ is non-singular if and only if operators  $\{N_b\}_b$ are linear independent. Now, if $m=d^2$ we can use the same reasoning to infer information completeness of $\mathbf{N}$. Finally, from the moment matrix $\Gamma$ we can directly verify the SIC property i.e.,~the condition $\tr(N_bN_{b'}) = \frac{1+\delta_{b,b'}}{d^2(d+1)}, \forall b,b'$.

\section{Discussion} We have presented a systematic analytical scheme for noise-resilient certification of overlaps between arbitrary configurations of pure quantum states and rank-$1$ projective measurements in the prepare-and-measure scenario. For qubits our scheme can be used to robustly self-test general ensembles of pure quantum sates and the corresponding projective measurements.  We believe that these findings pave the way towards systematic certification of general quantum systems in the semi-device-independent paradigm. This is supported by the concrete  qubit results from Table~\ref{tab:examples} and by the universality of our protocol for certifying overlaps of quantum sates in arbitrary dimension.

In our main claims, we state that the scheme can be applied to an arbitrary set of preparation states. However, the sets of states in Theorems~\ref{th:overlaps} and~\ref{th:rob} always come in a collection of bases. These two statements are not contradicting one another since any given constellation of states can always be completed to a set of bases by adding extra preparation states. In practice, one should aim at aligning each basis of preparation states and the corresponding measurement to the best degree allowed by the experimental setup.

There are many exciting research directions that stem from our work. One of them is a systematic SDI characterization of generalized measurements beyond qubits. Another research direction is the robustness analysis of our scheme in the presence of shared randomness. An important contribution to the field of self-testing would be a generalization of our self-testing results for qubit to higher-dimensional systems. Finally, it is interesting to combine our methods with recent findings relating state discrimination games and resource theories of measurements and channels  \cite{oszmaniec2019operational,TakagiPRX,Roope2019}.

We would also like to draw reader's attention to a related work~\cite{arminNew}.

\begin{acknowledgements}
We thank Marcin Paw\l owski, Jedrzej Kaniewski  and Tanmoy Biswas for interesting discussions and comments. NM acknowledges the financial support by First TEAM Grant No.~2016-1/5. We acknowledge partial support by the Foundation for Polish Science (IRAP project, ICTQT, contract no. 2018/MAB/5, co-financed by EU within Smart Growth Operational Programme). We acknowledge the support from International Research Agendas Programme of the Foundation for Polish Science funded by the  European Regional Development Fund. MO acknowledges the financial support by TEAM-NET project (contract no. POIR.04.04.00-00-17C1/18-00). 
 \end{acknowledgements}

\section*{Appendix}
\appendix

In this Appendix we provide technical details that were omitted in the main text. First, in Appendix~\ref{app:remark} we prove statements regarding saturation of the bounds in Theorem~\ref{th:overlaps}. Second, in Appendix~\ref{app:rob} we formulate and prove a quantitative version of Theorem \ref{th:rob_quant_fid}, which concerns robustness of our self-testing protocol expressed in terms of average fidelities. We also prove there Theorem~\ref{th:rob_quant} that gives analogous results expressed in terms of trace distance. In Appendix~\ref{app:procrust} we provide alternative robustness analysis based on orthogonal Procrustes problem. In Appendix~\ref{app:ex} we provide technical details about the considered examples in Table~\ref{tab:examples}. Finally, in Appendix~\ref{app:extra} we give proofs of two auxiliary Lemmas needed in the proof of the technical version of Theorem~\ref{th:rob}.

\section{Explicit form of states and measurements saturating the bounds in Theorem~\ref{th:overlaps}}
\label{app:remark}
We start by giving an example of four states and two qubit measurements ($n=2$, $d=2$), for which the bound in Eq.~\eqref{eq:overlaps} on the deviation of the overlaps scales like $\sqrt{\epsilon}$. Their explicit form is given below:
\beq
\rho^1_1 = \proj{\tilde{0}},\quad \ket{\tilde{0}} = \sqrt{1-\epsilon}\ket{0}+\sqrt{\epsilon}\ket{1},\quad \rho^2_1 = \proj{+},
\eeq
where $\ket{+} = \frac{1}{\sqrt{2}}(\ket{0}+\ket{1})$, and as required by our scheme $\rho^x_2 = \Id-\rho^x_1,\; x=1,2$ and $M^x_a = \rho^x_a$, for all $x,a$. Let the experimental states and measurements be the following:
\beq
&&\tilde{\rho}^1_1 = \proj{0},\quad \tilde{\rho}^1_2 = \proj{1},\quad \tilde{\rho}^2_1 = \proj{+},\quad \tilde{\rho}^2_2 = \proj{-},\nonumber\\
&& \tilde{M}^1_1 = \proj{\tilde{0}},\quad \tilde{M}^2_1 = \proj{\tilde{+}},\quad \ket{\tilde{+}} = \sqrt{1-\epsilon}\ket{+}+\sqrt{\epsilon}\ket{-}\ .
\eeq
First, we need to show that the experimental statistics is deviated from the target one by at most $\epsilon$, i.e.,~ $|\tr(\tilde{\rho}^x_a\tilde{M}^y_b)-\tr(\rho^x_a\rho^y_b)|\leq \epsilon, \forall a,b,x,y$. Since $\tilde{\rho}^x_1+\tilde{\rho}^x_2 = \Id$ for both $x=1,2$, it is sufficient to consider the cases of $a=1,b=1$. Calculating the statistics we can see that $\tr(\tilde{\rho}^1_1\tilde{M}^1_1)=\tr(\tilde{\rho}^2_1\tilde{M}^2_1) = 1-\epsilon$, and $\tr(\tilde{\rho}^1_1\tilde{M}^2_1)=\tr(\tilde{\rho}^2_1\tilde{M}^1_1) = \frac{1}{2}+\sqrt{\epsilon-\epsilon^2}$ and the latter is just the same as the target probability, $\tr(\rho^1_1\rho^2_1)$. Finally, calculation of the overlap between the states yields:
\beq
|\tr(\rho^1_1\rho^2_1)-\tr(\tilde{\rho}^1_1\tilde{\rho}^2_1)| = \sqrt{\epsilon-\epsilon^2},
\eeq
which confirms the scaling of the first order of $\sqrt{\epsilon}$.

Now, let us show that the bounds from Eq.~\eqref{eq:overlaps} are also tight in the first order of $d$. We still consider the case of $n=2$, but now $d$ is arbitrary. Let us consider the following POVM: 
\begin{equation}
M_1 = \proj{1}+\epsilon(d-1)\proj{+}\ ,\ M_a = (\ket{a}-\delta\ket{+})(\bra{a}-\delta\bra{+})\ a=2,3,\dots,d\ ,
\end{equation}
where $\delta =\frac{1}{\sqrt{d-1}}+\sqrt{\frac{1}{d-1}-\epsilon}$,  $\{\ket{a}\}_{a=1}^d$ is the computational basis of $\Cl^d$  and $\ket{+} = \frac{1}{\sqrt{d-1}}\sum_{a=2}^{d}\ket{a}$ is the "maximally coherent" state in the subspace spanned by $\{\ket{i}\}_{i=2}^d$. In the proof of Theorem~\ref{th:overlaps} the quantity $|\tr(\tilde{\rho}^x_a\tilde{\rho}^{x'}_{a'})-\tr(\rho^x_a\rho^{x'}_{a'})|$ is upper-bounded by $\epsilon+\norm{\tilde{\rho}^x_a-\tilde{M}^x_a}$, for which there exist a state $\tilde{\rho}^{x'}_{a'}$ and an effect $\tilde{M}^{x'}_{a'}$ reaching the bound. Now if we take $\Mbt^x$ to be the POVM we just introduced, and $\tilde{\rho}^x_1=\proj{1}$, $\tilde{\rho}^x_a = \frac{M_a}{\tr(M_a)}$, $a=2,3,\dots,d$, the conditions of the Theorem~\ref{th:overlaps} will be satisfied and the resulting bound will be $d\epsilon$.

\section{Quantitative statement of the robust self-testing }
\label{app:rob}
In this part we first formulate Theorem~\ref{th:rob_quant} in which 
we present robustness of our self-testing scheme in terms of trace distance\footnote{Recall that the trace distance is defined via $\td{\sigma}{\rho} = \frac{1}{2}\norm{\sigma-\rho}_1  = \frac{1}{2}\tr\sqrt{(\sigma-\rho)^\dagger(\tilde{\rho}-\rho)}$ and it has a neat operational interpretation in term of optimal success probability $p$ of distinguishing $\rho$ and $\sigma$ via the most general quantum measurements: $p=\frac{1}{2}(1+\td{\sigma}{\rho})$.  } for states and operator norm for measurements.  Then, we give a quantitative statement of  Theorem~\ref{th:rob}, that concerns robustness of our self-testing scheme in terms of average fidelities. We then proceed with proofs of both results.

In what follows we will need the following definition.
\begin{definition}
Let $\{\mathbf{n}_i\}_{i\in [n]}$ be a set of $n$ vectors in $\Rl^l$. Matrix $\Gamma \in \Rl^{n\times n}$ is called the Gram matrix of the set $\{\mathbf{n}_i\}_{i\in [n]}$, if its elements are given by $\Gamma_{i,j} = \mathbf{n}_i\cdot\mathbf{n}_j$, $i,j\in [n]$.
\end{definition}

\setcounter{theorem}{2}
\begin{theorem}[Robust self-testing for qubits via trace distance and operator norm]
\label{th:rob_quant}
Consider  pure target qubit preparation states $\rho^x_a$ and target projective measurements $\Mb^y$,  where $a=1,2$ and $x,y\in [n]$.  Assume that $\rho^x_a = M^x_a$ for all $a,x$ and, furthermore, that experimental states $\tilde{\rho}^x_a$ and measurements $\Mbt^y$ act on Hilbert space of dimension at most $d$ and generate statistics $\tilde{p}(b|a,x,y)=\tr(\tilde{\rho}^x_a\tilde{M}^y_b)$ such that $|\tilde{p}(b|a,x,y)- \tr(\rho^x_a M^y_b)|\leq \epsilon$, for all  $a,b,x,y$. 

Let $k\in\lbrace{2,3\rbrace}$ be the cardinality of the maximal set of linearly independent  Bloch vectors of states $\rho^x_1$. Fix a set $S \subset [n]$ of $k$ linearly independent vectors $\{\mathbf{n}^x_1\}_{x\in S}$, construct their Gram matrix $\Gamma_S$, and let $L_S$ be a Cholesky factor of 
$\Gamma_S$ (i.e.,~$\Gamma_S = L_SL^T_S$ and $L_S$ is lower triangular). For every $x \in [n]\setminus S$ and both $a=1,2$ let  $\mathbf{c}^{x,a}_S$, denote the coefficients of decomposition of $\mathbf{n}^x_a$ as a linear combination of $\{\mathbf{n}^x_1\}_{x\in S}$. Finally, let us define three auxiliary functions
\beq
&& F_k(\epsilon)\coloneqq \sqrt{\epsilon}\sqrt{4k(k-1)}\sqrt{1+2\sqrt{\epsilon}+\frac{k+3}{k-1}\epsilon}\ , \nonumber\\
&& O_k(\epsilon)\coloneqq 2((k-1)\sqrt{\epsilon}+(k+1)\epsilon)\ ,\label{eq:rob_E} \\
&& E_{S,k}(\epsilon) \coloneqq \frac{1}{2\sqrt{2}}\frac{\norm{\Gamma_S^{-1}}F_k(\epsilon)}{\sqrt{1-\norm{\Gamma_S^{-1}}O_k(\epsilon)}}\min\left[\frac{\norm{\Gamma_S}}{\norm{\Gamma_S}_F},\frac{\norm{L_S}}{\sqrt{k}}\right]\ . \nonumber
\eeq
Then, there is a region of $\epsilon\in[0,\epsilon_0)$ determined by
\beq
\label{eq:rob_cond}
\norm{\Gamma_S^{-1}}O_k(\epsilon)\leq 1,
\eeq
for which there exist a qubit unitary matrix $U$  such that
\begin{subequations}\label{eq:rob_bounds}
\begin{align}
\label{eq:rob_1}
&\frac{1}{k}\sum_{x\in S}\td{U(\tilde{\rho}^x_1)^{(T)} U^\dagger}{\rho^x_1}\leq E_{S,k}(\epsilon)\ ,\\
\label{eq:rob_2}
&\frac{1}{k}\sum_{x\in S}\td{U(\tilde{\rho}^x_2)^{(T)} U^\dagger}{\rho^x_2}\leq E_{S,k}(\epsilon)+2\sqrt{\epsilon}\ ,\\
\label{eq:rob_3}
&\td{U(\tilde{\rho}^x_a)^{(T)} U^\dagger}{\rho^x_a} \leq \sqrt{k}{|\mathbf{c}^{a,x}_S|E_{S,k}(\epsilon)}+\frac{\sqrt{k}}{2}\left(|\mathbf{c}^{a,x}_S|+\frac{\sqrt{k}}{k-1}\right)\frac{\norm{\Gamma_S^{-1}}O_k(\epsilon)}{1-\norm{\Gamma_S^{-1}}O_k(\epsilon)},\\
&\text{for}\; x\notin S,\ a=1,2\ ,\nonumber\\
\label{eq:rob_m_1}
&\frac{1}{k}\sum_{y\in S}\norm{U(\tilde{M}^y_1)^{(T)}U^\dagger-M^y_1}\leq E_{S,k}(\epsilon) + \sqrt{\epsilon}\ ,\\
\label{eq:rob_m_2}
&\norm{U(\tilde{M}^y_1)^{(T)}U^\dagger-M^y_1}\leq \td{U(\tilde{\rho}^y_1)^{(T)} U^\dagger}{\rho^y_1} + \sqrt{\epsilon},\quad \text{for}\; y\notin S\  ,
\end{align}
\end{subequations}
where  $(\cdot)^{(T)}$ is the transposition (with respect a fixed basis in $\Cl^2$) that may have to be applied to all experimental states and measurements simultaneously. 
\end{theorem}
The reason for such a formulation of the theorem lays in the proof techniques that were used to derive it. Namely, we use results on the stability of the Cholesky factorization. Since the result employed by us (Ref.~\cite{sun1991perturbation}) is valid for positive-definite matrices we cannot apply it to the Gram matrix of all of the vectors $\mathbf{n}^x_a$ due to their linear dependence.  As a result, our bounds depend on the particular choice of the set $S$ of states whose Bloch vectors are linearly independent. In what follows, without the loss of generality we take $S=\{1,2\}$ for $k=2$ and $S=\{1,2,3\}$ for $k=3$ in the proof.  In practice, however, one should consider all subsets $S$ of states that give non-singular Gram matrices having in mind a specific application. 

\emph{Outline of the proof.--} In what follows we will consider a particular subset $\{\mathbf{n}^x_1\}_{x\in S}$ of $k$ linearly independent Bloch vectors. For simplify we will omit the subscript $S$ whenever possible. The main idea of the proof is to use a Cholesky factors of $k\times k$ the Gram matrices 
$\Gamma$ and $\tilde{\Gamma}$ of target ($\{\mathbf{n}^x_1\}_{x\in S}$) and experimental ($\{\mathbf{\tilde{n}}^x_1\}_{x\in S}$) Bloch vectors respectively.  We then make use of the result of Ref.~\cite{sun1991perturbation} on the stability of Cholesky factorization which, in simple terms, states that if $\Gamma$ and $\tilde{\Gamma}$ are close to each other, so are their Cholesky factors $L$ and $\tilde{L}$. Specifically, this result, which we quote bellow (see Theorem [Sun 1991]), sets an upper bound on the Frobenius norm of $\Delta L = \tilde{L}-L$ in terms of the Frobenius norm of the perturbation  $\Delta\Gamma = \tilde{\Gamma}-\Gamma$. The Frobeniuous norm of the perturbation $\Delta L$ can be connected to the trace distance between states $\rho^x_1$ and the rotated states $\tilde{\rho}^x_1$ in the selected subset $S$. On the other hand, the bound on $\norm{\Delta\Gamma}_F$ can be estimated from our assumption: $|p(b|a,x,y)-\tilde{p}(b|a,x,y)|\leq \epsilon$. This bound follows directly from the results of Theorem~\ref{th:overlaps}, but here we use an improved qubit stability bounds, which are given by Lemma~\ref{lemma:overlaps}. This, in turn, leads to stronger estimates for the norms of $\Delta\Gamma$ in Lemma~\ref{lemma:norm} bellow. Combining all these results produces bounds in Eq.~\eqref{eq:rob_1}. 

In the next part of the proof we determine the trace distance between states $\rho^x_2$ and $\tilde{\rho}^x_2$ for $x\in S$ based solely on the fact that $\rho^x_2 = \Id-\rho^x_1$ and $\tilde{\rho}^x_2$ are \emph{close to} $\Id-\tilde{\rho}^x_1$. This gives the bounds in Eq.~\eqref{eq:rob_2}. For states $x\notin S$ we use the fact that they can be decomposed in the basis of the states in $S$. Since this linear decomposition can be different for target and experimental states, we use the result on stability of linear systems (see Theorem [Higham 2002], which we also quote below). The bounds for the states $x\notin S$ are given by Eq.~\eqref{eq:rob_3}. Finally, we connect the bounds for the distance between the target and experimental measurements and the distance between the corresponding states resulting in Eq.~\eqref{eq:rob_m_1} and \eqref{eq:rob_m_2}.

\setcounter{theorem}{1}
\begin{theorem}[Quantitative formulation of robust self-testing for qubits]
\label{th:rob_quant_fid}
Consider  pure target qubit preparation states $\rho^x_a$ and target projective measurements $\Mb^y$,  where $a=1,2$ and $x,y\in [n]$.  Assume that $\rho^x_a = M^x_a$ for all $a,x$ and, furthermore, that experimental states $\tilde{\rho}^x_a$ and measurements $\Mbt^y$ act on Hilbert space of dimension at most $d$ and generate statistics $\tilde{p}(b|a,x,y)=\tr(\tilde{\rho}^x_a\tilde{M}^y_b)$ such that $|\tilde{p}(b|a,x,y)- \tr(\rho^x_a M^y_b)|\leq \epsilon$, for all  $a,b,x,y$. 

Let $k\in\lbrace{2,3\rbrace}$ be the cardinality of the maximal set of linearly independent  Bloch vectors of states $\rho^x_1$. Fix a set $S \subset [n]$ of $k$ linearly independent vectors $\{\mathbf{n}^x_1\}_{x\in S}$, construct their Gram matrix $\Gamma_S$, and let $L_S$ be a Cholesky factor of 
$\Gamma_S$ (i.e.,~$\Gamma_S = L_SL^T_S$ and $L_S$ is lower triangular). For every $x \in [n]\setminus S$ and both $a=1,2$ let  $\mathbf{c}^{x,a}_S$, denote the coefficients of decomposition of $\mathbf{n}^x_a$ as a linear combination of $\{\mathbf{n}^x_1\}_{x\in S}$. Finally, let us define three auxiliary functions
\beq
&& F_k(\epsilon)\coloneqq \sqrt{\epsilon}\sqrt{4k(k-1)}\sqrt{1+2\sqrt{\epsilon}+\frac{k+3}{k-1}\epsilon}\ , \nonumber\\
&& O_k(\epsilon)\coloneqq 2((k-1)\sqrt{\epsilon}+(k+1)\epsilon)\ , \nonumber\\
&& E_{S,k}(\epsilon) \coloneqq \frac{1}{2\sqrt{2}}\frac{\norm{\Gamma_S^{-1}}F_k(\epsilon)}{\sqrt{1-\norm{\Gamma_S^{-1}}O_k(\epsilon)}}\min\left[\frac{\norm{\Gamma_S}}{\norm{\Gamma_S}_F},\frac{\norm{L_S}}{\sqrt{k}}\right]\ . \label{eq:rob_E2}
\eeq
Then, there is a region of $\epsilon\in[0,\epsilon_0)$ determined by
\beq
\label{eq:rob_cond2}
\norm{\Gamma_S^{-1}}O_k(\epsilon)\leq 1,
\eeq
for which there exist a qubit unitary matrix $U$  such that 
\begin{subequations}\label{eq:rob_bounds2}
\begin{align}
\label{eq:rob_1_2}
&\frac{1}{k}\sum_{x\in S}\tr(U(\tilde{\rho}^x_1)^{(T)} U^\dagger \rho^x_1 )\geq 1- \frac{\epsilon(1-2\epsilon)}{(1-\epsilon)^2} - E_{S,k}(\epsilon)^2\ ,\\
\label{eq:rob_2_2}
&\frac{1}{k}\sum_{x\in S}\tr(U(\tilde{\rho}^x_2)^{(T)} U^\dagger \rho^x_2)\geq 1- \frac{\epsilon(1-2\epsilon)}{(1-\epsilon)^2} - (E_{S,k}(\epsilon)+2\sqrt{\epsilon})^2\ ,\\
\label{eq:rob_3_2}
&\tr(U(\tilde{\rho}^x_a)^{(T)} U^\dagger \rho^x_a) \geq 1-\frac{\epsilon(1-2\epsilon)}{(1-\epsilon)^2}-\left(\td{U(\tilde{\rho}^x_a)^{(T)} U^\dagger}{\rho^x_a} \right)^2,\quad \text{for}\; x\notin S,\ a=1,2\ ,\\
\label{eq:rob_m_1_2}
&\frac{1}{k}\sum_{y\in S}\tr( U(\tilde{M}^y_1)^{(T)}U^\dagger M^y_1 ) \geq 1 - \frac{5}{2}\epsilon -\left(\sqrt{2}E_{S,k}(\epsilon) +\sqrt{\epsilon}\right)^2   \ ,\\
\label{eq:rob_m_2_2}
&\tr( U(\tilde{M}^y_1)^{(T)}U^\dagger M^y_1 )
 \geq 1 -\epsilon - \left(\td{U(\tilde{\rho}^y_1)^{(T)} U^\dagger}{\rho^y_1} +\sqrt{\epsilon}\right)^2,\quad \text{for}\; y\notin S\  ,
\end{align}
\end{subequations}
where  $(\cdot)^{(T)}$ is the transposition (with respect a fixed basis in $\Cl^2$) that may have to be applied to all experimental states and measurements simultaneously. Note that in formulas \eqref{eq:rob_3_2} and \eqref{eq:rob_m_2} we have used trace distance $\td{U(\tilde{\rho}^x_a)^{(T)} U^\dagger}{\rho^x_a}$ in order to simplify the resulting formulas. The bound on this quantity is given in Eq.~\eqref{eq:rob_3}.
\end{theorem}

\begin{remark}
Results of Theorem~\ref{th:rob_quant_fid} are obtained using a similar reasoning to that given in the outline of the proof of Theorem \ref{th:rob_quant}. However, the proof steps are supplemented by bounds connecting fidelity to the trace distance (for states) and operator norm (for measurement operators). The proof of Theorem~\ref{th:rob_quant_fid} is given at the end of this part of the Appendix. 
\end{remark} 

Before we proceed we state two theorems from the literature that are needed for our proof. First, we repeat the statement of the Theorem~1.4 from Ref.~\cite{sun1991perturbation} for real-valued matrices. Second, we state a result concerning stability of systems of linear equation, which we borrow from Ref.~\cite{higham2002accuracy} (Theorem 7.2).  We changed the notation used in these papers in accordance with our paper.

\begin{theorem*}[Sun 1991 - Stability of Cholesky factorisation]
Let $\Gamma$ be an $k\times k$ positive definite matrix and $\Gamma = LL^T$ its Cholesky factorization. If $\Delta\Gamma$ is an $k\times k$ symmetric matrix satisfying 
\beq
\label{eq:ch_bound_condition}
\norm{\Gamma^{-1}}\norm{\Delta\Gamma}\leq 1,
\eeq
then there is a unique Cholesky factorization
\beq
\Gamma+\Delta\Gamma = (L+\Delta L)(L+\Delta L)^T,\nonumber
\eeq
and
\beq
\label{eq:ch_bound}
\norm{\Delta L}_F\leq \frac{\norm{\Gamma^{-1}}\norm{\Delta\Gamma}_F}{\sqrt{2(1-\norm{\Gamma^{-1}}\norm{\Delta\Gamma}})}\min\left[\frac{\norm{L}_F\norm{\Gamma}}{\norm{\Gamma}_F},\norm{L}\right].
\eeq
\end{theorem*}

\begin{remark}
This theorem dates back to 1991, and, of course there have been attempts to improve this result. However, the recent review \cite{chang2010rigorous} on this topic suggests that the bound given in the Theorem above is the most appropriate in our case (Remark 3.2, Ref.~\cite{chang2010rigorous}).
\end{remark}

\begin{theorem*}[Higham 2002 - Stability of systems of linear equations]
Let $\mathbf{c}$ be a solution of a system of linear equations $\Gamma\mathbf{c} = \mathbf{g}$, where $\mathbf{c},\mathbf{g}\in\Rl^k$ and $\Gamma\in\mathrm{Mat}_{k\times k}(\Rl)$. Let now $\mathbf{\tilde{c}}$ be a solution to  $(\Gamma+\Delta\Gamma)\mathbf{\tilde{c}} = \mathbf{g}+\mathbf{\Delta g}$, where $\Delta\Gamma\in\mathrm{Mat}_{k\times k}(\Rl) $, $\mathbf{\Delta g}\in\Rl^k$. Assume that there exits $\delta'>0$ such that
$\norm{\Delta\Gamma}\leq \delta'\norm{E}$ and $|\mathbf{\Delta g}|\leq \delta' |\mathbf{f}|$ (for some $E\in\mathrm{Mat}_{k\times k}(\Rl)$, $\mathbf{f}\in \Rl^k$), and that $\delta'\norm{\Gamma^{-1}}\norm{E}\leq 1$. Then, we have 
\beq
\frac{|\mathbf{c}-\mathbf{\tilde{c}}|}{|\mathbf{c}\,|}\leq \frac{\delta'}{1-\delta'\norm{\Gamma^{-1}}\norm{E}}\left(\frac{\norm{\Gamma^{-1}}|\mathbf{f}|}{|\mathbf{c}\,|}+\norm{\Gamma^{-1}}\norm{E}\right)\ . 
\eeq
\end{theorem*}

\begin{lemma}
\label{lemma:overlaps}
Under the conditions of Theorem~\ref{th:overlaps} qubit states and measurements  $\{\rho^x_a\}_{a,x}$, $\{\Mb^y\}_{y}$ and $\{\tilde{\rho}^x_a\}_{a,x}$, $\{\Mbt^y\}_{y}$ satisfy
\beq
&&\norm{\tilde{\rho}^x_a} \geq \frac{1-2\epsilon}{1-\epsilon},\; \forall a,x,\nonumber\\
&&|\tr(\tilde{\rho}^x_a\tilde{\rho}^{x'}_{a'})-\tr(\rho^x_a\rho^{x'}_{a'})|\leq \epsilon+\sqrt{\epsilon},\nonumber\\
&&|\tr(\tilde{M}^y_b\tilde{M}^{y'}_{b'})-\tr(M^y_bM^{y'}_{b'})|\leq \epsilon+(1+\epsilon)\sqrt{\epsilon},\nonumber
\eeq
$\forall x\neq x', a\neq a'$ and $\forall y\neq y', b\neq b'$ respectively, and $\epsilon \leq \frac{1}{3}$.
\end{lemma}

\begin{lemma}
\label{lemma:norm}
Let $\Gamma$ be the Gram matrix of Bloch vectors of target states $\{\rho^x_1\}_{x\in S}$. Likewise, let $\tilde{\Gamma}$ be the Gram matrix of Bloch vectors of experimental states $\{\tilde{\rho}^x_1\}_{x\in S}$. Let  $|p(b|a,x,y)-\tilde{p}(b|a,x,y)|\leq \epsilon$ for all $b,a,x,y$ and let $k=|S|$ be the cardinality of the set $S$. Then,  we have the following bounds on norms of $\Delta\Gamma = \tilde{\Gamma}-\Gamma$ 
\beq
 \norm{\Delta\Gamma}_F \leq  F_k(\epsilon)\ ,\nonumber \\ 
 \norm{\Delta\Gamma} \leq O_k(\epsilon)\ ,
\eeq
where 
\beq
&& F_k(\epsilon) = \sqrt{\epsilon}\sqrt{4k(k-1)}\sqrt{1+2\sqrt{\epsilon}+\frac{k+3}{k-1}\epsilon}\ , \nonumber\\
&& O_k(\epsilon) = 2((k-1)\sqrt{\epsilon}+(k+1)\epsilon)\ . \nonumber
\eeq
\end{lemma}

\begin{proof}[Proof of Theorem~\ref{th:rob_quant}] Let $\tilde{\Gamma}$ be the Gram matrix of Bloch vectors $\{\mathbf{\tilde{n}}^x_1\}_{x\in S}$ of the experimental states for $x\in S$. Let us for now assume that the Cholesky factorization $\tilde{\Gamma}=\tilde{L}\tilde{L}^T$ exists and let $\Delta L = \tilde{L}-L$. 

We start our \emph{first part of the proof} regarding the states $x\in S$, $a=1$ by connecting the square of the norm $\norm{\Delta L}_F$ with trace distance between the states $\{\tilde{\rho}^x_1\}_{x\in S}$ and $\{\rho^x_1\}_{x\in S}$. From the sketch of the proof of Theorem~\ref{th:rob_quant} it follows that the Bloch vectors $\{\vn^x_1\}_{x\in S}$ of the target states and the vectors $\{\mathbf{l}^x\}_{x\in S}$, transpose of which are rows of $L$, are connected via an orthogonal transformation, which we denote as $O$, i.e.,~$\mathbf{l}^x = O\vn^x_1$, $x\in S$. Similarly, for the Bloch vectors $\{\vnt^x_1\}_{x\in S}$ of the experimental states and the vectors $\{\mathbf{\tilde{l}}^x\}_{x\in S}$ that form $\tilde{L}$ there exists an orthogonal transformation $\tilde{O}$, such that $\mathbf{\tilde{l}}^x = \tilde{O}\vnt^x_1$, $x\in S$. Let us define states $\tau^x = \frac{1}{2}(\Id + \mathbf{l}^x\cdot\bm{\sigma})$, and $\tilde{\tau}^x = \frac{1}{2}(\Id + \mathbf{\tilde{l}}^x\cdot\bm{\sigma})$, $x\in S$. We know that there exist unitary transformations $V$ and $\tilde{V}$ such that: $\tau^x = V(\rho^x_1)^{(T)}V^\dagger$, and $\tilde{\tau}^x = \tilde{V}(\tilde{\rho}^x_1)^{(T)}\tilde{V}^\dagger$ for $x\in S$ (the optional transposition $(\cdot)^{(T)}$ is reserved for the cases $\det(O) = -1$ and $\det(\tilde{O}) = -1$). We have the following sequence of equalities valid for $x\in S$:
\beq
\label{eq:ap_fid_u}
\td{\tilde{\tau}^x}{\tau^x} &=& \td{\tilde{V}(\tilde{\rho}^x_1)^{(T)}\tilde{V}^\dagger}{V(\rho^x_1)^{(T)}V^\dagger} = \td{V^\dagger\tilde{V}(\tilde{\rho}^x_1)^{(T)}\tilde{V}^\dagger V}{(\rho^x_1)^{(T)}} \nonumber\\
&=& \td{U(\tilde{\rho}^x_1)^{(T)}U^\dagger}{\rho^x_1}\ ,
\eeq
where we used the invariance of trace distancee under unitary evolution and transposition, and also denoted the resulting unitary $V^\dagger\tilde{V}$ as $U$. It is clear that one transposition in the final formula in Eq.~\eqref{eq:ap_fid_u} is enough as the case $\det(\tilde{O}) = \det(O) = -1$ is equivalent to having $\det(\tilde{O}) = \det(O) = 1$ and changing $U$ to $U^T$. 

We need to show now that this unitary $U$ satisfies the claim of Theorem~\ref{th:rob} for the states $x\in S$. For qubits, trace distance can be expressed directly via Eucleedian distance between Bloch vectors:
\beq
\label{eq:ap_fid_tr}
\td{\tilde{\tau}^x}{\tau^x} = \frac{1}{2}|\mathbf{\tilde{l}}^x-\mathbf{l}^x|\ .
\eeq
Using this fact and Eq.~\eqref{eq:ap_fid_u}, it is straightforward to derive the following upper bound:
\beq
\label{eq:ap_fid_bound_3}
\frac{1}{k}\sum_{x=1}^k\td{U(\tilde{\rho}^x_1)^{(T)}U^\dagger}{\rho^x_1}\leq \frac{1}{2\sqrt{k}}\norm{\Delta L}_F\ .
\eeq
We now use Theorem [Sun 1991] (Ref.~\cite{sun1991perturbation}) to upper-bound $\norm{\Delta L}_F$. Specifically,  we apply it to the Gram matrix of the Bloch vectors of states $\{\rho^x_1\}_{x\in S}$. We have $\norm{L}_F = \sqrt{k}$, since the target states are assumed to be pure. The direct substitution of the bound in Eq.~\eqref{eq:ch_bound} into Eq.~\eqref{eq:ap_fid_bound_3} gives the bound in Eq.~\eqref{eq:rob_1}, and also the condition in Eq.~(\ref{eq:rob_cond}) when the Theorem~\ref{th:rob_quant} (and also Theorem~\ref{th:rob}) applies. 

In the beginning of the proof we assumed that the Cholesky factorization $\tilde{\Gamma} = \tilde{L}\tilde{L}^T$ exists. The condition in Eq.~\eqref{eq:rob_cond} gives the sufficient condition for this to hold. We conclude this part of the proof by noting that a similar statement in terms of fidelity (Theorem~\ref{th:rob_quant_fid}, Eq.~\eqref{eq:rob_1_2}) follows from Eq.~\eqref{eq:ap_fid_bound_3} by connecting fidelity to the trace distance as described in Eq.~\eqref{eq:final_fid_bound}.

In \emph{the second part of the proof} we derive upper bounds on the trace distances between the states corresponding to $x\in S$ and $a=2$. As we will see, those can be connected to the bounds for the states with $x\in S$ and $a=1$. Indeed, we can write the following for every $x\in S$:
\beq
\label{eq:ap_fid_a2}
2\td{U\tilde{\rho}^x_2U^\dagger}{\rho^x_2} \hspace{-12pt}&&= \norm{\rho^x_2-U\tilde{\rho}^x_2U^\dagger}_1\leq \norm{\Id-\rho^x_1-U(\Id-\tilde{\rho}^x_1)U^\dagger}_1+\norm{U(\Id-\tilde{\rho}^x_1-\tilde{\rho}^x_2)U^\dagger}_1\nonumber\\
&&\leq \norm{\rho^x_1-U\tilde{\rho}^x_1U^\dagger}_1+\norm{\Id-\tilde{\rho}^x_1+\tilde{M}^x_2}_1-\norm{\tilde{M}^x_2-\tilde{\rho}^x_2}_1\\
&&= 2\td{U\tilde{\rho}^x_1U^\dagger}{\rho^x_1}+\sum_{a=1,2}\norm{\tilde{\rho}^x_a-\tilde{M}^x_a}_1.\nonumber
\eeq
Exactly the same reasoning can be applied to upper bound the trace distance $\td{U(\tilde{\rho}^x_2)^{T}U^\dagger}{\rho^x_2}$, in which case the transposition will propagate to $\td{U(\tilde{\rho}^x_1)^TU^\dagger}{\rho^x_1}$ but would not affect the term $\sum_{a=1,2}\norm{\tilde{\rho}^x_a-\tilde{M}^x_a}_1$ (as we can take $(\tilde{M}^x_a)^T$). 

To finish the calculations we need to upper-bound the second summand in Eq.~(\ref{eq:ap_fid_a2}). We do it by writing $\norm{\tilde{\rho}^x_a-\tilde{M}^x_a}_1 \leq 2\norm{\tilde{\rho}^x_a-\tilde{M}^x_a} \leq 2\sqrt{\epsilon}$, $\forall a,x$, where the last inequality is proven in Lemma~\ref{lemma:overlaps} (see Eq.~(\ref{app:eq_refined_norm_bound})). From here, it is easy to obtain the bound in Eq.~(\ref{eq:rob_2}). The reason why we do not simply apply the same reasoning to the states corresponding to $a=2$ and $x\in S$ as for the states with $a=1$ is because we want the isometry $U$ to be the same for all of the states.

In \emph{the third part of the proof} we derive the bounds for the preparations states corresponding to $x\notin S$ and both $a=1,2$. We start by reminding ourselves that the set of states $\{\rho^x_1\}_{x\in S}$ is assumed to be tomographically complete. Is is equivalent to the assumption that the vectors $\{\mathbf{l}^x\}_{x\in S}$, defined above, are linearly independent and span $\Rl^3$ for $k=3$, or the considered subspace for $k=2$. The condition in Eq.~(\ref{eq:rob_cond}) of the Theorem~\ref{th:rob_quant} (also stated in Eq.~(\ref{eq:ch_bound_condition})) ensures that the same holds for the vectors $\{\mathbf{\tilde{l}}^x\}_{x\in S}$. If so, let us expand the Bloch vectors of the states $\rho^{x'}_{a'}$ and $U\tilde{\rho}^{x'}_{a'}U^\dagger$, $x'\notin S$ in terms of $\{\mathbf{l}^x\}_{x\in S}$ and $\{\mathbf{\tilde{l}}^x\}_{x\in S}$ respectively. Let us denote the coefficients of these linear expansions as $\{c^{a',x'}_x\}_{x\in S}$ and $\{\tilde{c}^{a',x'}_x\}_{x\in S}$ respectively. These coefficients will, of course, depend on $x'$ and $a'$ as well as on the choice of the set $S$. However, for simplicity of the derivations we are going omit the subscripts $x'$ and $a'$ until we present the final result. We will also assume, without loss of generality, that $S=\{1,2,3\}$ (or $S=\{1,2\}$ for $k=2$).

It is clear that the coefficients $\{c_x\}_{x\in S}$ and $\{\tilde{c}_x\}_{x\in S}$ satisfy the following respective systems of linear equations:
\beq
\Gamma\mathbf{c} = \mathbf{g},\quad \tilde{\Gamma}\mathbf{\tilde{c}} = \mathbf{\tilde{g}},
\eeq 
where $\mathbf{c} = (c_1,c_2,c_3)$, $\mathbf{g} = (2\tr(\rho^{x'}_{a'}\rho^1_1)-1,2\tr(\rho^{x'}_{a'}\rho^2_1)-1,2\tr(\rho^{x'}_{a'}\rho^3_1)-1)$ and analogously $\mathbf{\tilde{c}} = (\tilde{c}_1,\tilde{c}_2,\tilde{c}_3)$, $\mathbf{g} = (2\tr(\tilde{\rho}^{x'}_{a'}\tilde{\rho}^1_1)-1,2\tr(\tilde{\rho}^{x'}_{a'}\tilde{\rho}^2_1)-1,2\tr(\tilde{\rho}^{x'}_{a'}\tilde{\rho}^3_1)-1)$ whenever the cardinality $k$ of the set $S$ is $3$. For $k=2$ and $\mathbf{c},\mathbf{\tilde{c}},\mathbf{g},\mathbf{\tilde{g}}\in \Rl^2$ their definition is analogous. We again omitted the subscripts $a',x'$ for the vectors $\mathbf{g}$ and $\mathbf{\tilde{g}}$ for simplicity. The matrices $\Gamma$ and $\tilde{\Gamma}$ are still the moment matrices for the states $x\in S$, $a=1$.

In complete analogy to Eq.~(\ref{eq:ap_fid_tr}) we can write:
\beq
\td{U(\tilde{\rho}^{x'}_{a'})^{(T)}U^\dagger}{\rho^{x'}_{a'}}  = \frac{1}{2}\Big|\sum_{x\in S}(\tilde{c}_x\mathbf{\tilde{l}}^x-c_x\mathbf{l}^x)\Big|.
\eeq
We then can upper-bound the latter norm as follows:
\beq
\label{eq:ap_euc_exp}
&&\hspace{-12pt}\Big|\sum_{x\in S}(\tilde{c}_x\mathbf{\tilde{l}}^x-c_x\mathbf{l}^x)\Big|\leq \Big|\sum_{x\in S}c_x(\mathbf{\tilde{l}}^x-\mathbf{l}^x)\Big|+\Big|\sum_{x\in S}(\tilde{c}_x-c_x)\mathbf{\tilde{l}}^x\Big|\leq |\mathbf{c}\,|\norm{\Delta L}_F + \sum_{x\in S}|\tilde{c}_x-c_x||\mathbf{\tilde{l}}^x|.\nonumber\\&&
\eeq
In the equation above we used the following relation:
\beq
|\sum_{x\in S}c_x(\mathbf{\tilde{l}}^x-\mathbf{l}^x)| = \sqrt{\sum_{i=1}^3\left(\sum_{x\in S}c_x(\tilde{l}^x_i-l^x_i)\right)^2}\leq \sqrt{\sum_{x'\in S}c^2_{x'}\sum_{i=1}^3\sum_{x\in S}(\tilde{l}^x_i-l^x_i)^2} = |\mathbf{c}\,|\norm{\Delta L}_F.
\eeq
Also, since the norms $|\mathbf{\tilde{l}}^x|$ can be upper-bounded by $1$ for all $x$, the resulting upper bound can be simplified further to be $|\mathbf{c}\,|\norm{\Delta L}_F + \sum_{x\in S}|\tilde{c}_x-c_x|$.

To estimate the deviation $|\tilde{c}_x-c_x|$ we apply Theorem [Higham 2002] by taking $\Delta\Gamma  = \tilde{\Gamma}-\Gamma$ and $\mathbf{\Delta g} = \mathbf{\tilde{g}}-\mathbf{g}$. The bound on the operator norm $\norm{\Delta\Gamma}$ is given by Lemma~\ref{lemma:norm}. At the same time $|\mathbf{\Delta g}| = 2\sqrt{\sum_{x\in S}(\tr(\rho^{x'}_{a'}\rho^x_1)-\tr(\tilde{\rho}^{x'}_{a'}\tilde{\rho}^x_1))^2}\leq 2\sqrt{k}(\sqrt{\epsilon}+\epsilon)$. If we take $\delta' = 2((k-1)\sqrt{\epsilon}+(k+1)\epsilon)$, some matrix $E$ with $\norm{E}=1$, vector $\mathbf{f}$ such that $|\mathbf{f}| = \frac{\sqrt{k}}{k-1}$, we satisfy the conditions of the Theorem [Higham 2002] and get the following bound:
\beq
\label{eq:ap_stab_lin}
|\mathbf{c}-\mathbf{\tilde{c}}|\leq \left(|\mathbf{c}|+\frac{\sqrt{k}}{k-1}\right)\frac{\norm{\Gamma^{-1}}\norm{\Delta\Gamma}}{1-\norm{\Gamma^{-1}}\norm{\Delta\Gamma}}\ ,
\eeq
As the final step we need to connect the bound in Eq.~(\ref{eq:ap_euc_exp}) with the one in Eq.~(\ref{eq:ap_stab_lin}) by the relation between $1$-norm and the Euclidean norm in $\Rl^k$, which effectively adds a factor of $\sqrt{k}$ to the bound in Eq.~(\ref{eq:ap_stab_lin}). Combining everything together we obtain the following:
\beq
&&\hspace{-12pt}\td{U(\tilde{\rho}^{x'}_{a'})^{(T)}U^\dagger}{\rho^{x'}_{a'}} \leq \frac{1}{2}|\mathbf{c}\,|\norm{\Delta L}_F+\frac{\sqrt{k}}{2}\left(|\mathbf{c}|+\frac{\sqrt{k}}{k-1}\right)\frac{\norm{\Gamma^{-1}}\norm{\Delta\Gamma}}{1-\norm{\Gamma^{-1}}\norm{\Delta\Gamma}},\nonumber\\
&& x'\notin S, a'=1,2.
\eeq 
The above bound is given in Eq.~(\ref{eq:rob_3}) in terms of the quantity $E_{S,k}(\epsilon)$ (Eq.~(\ref{eq:rob_E})), where we took $\norm{\Delta L_S}_F= 2\sqrt{k}E_{S,k}(\epsilon)$ and brought back all the necessary subscripts. Again, using the relation in Eq.~(\ref{eq:tr_dist_fid_relation}) we can derive bounds on the fidelity between the states $\rho^x_a$ and $\tilde{\rho}^x_a$ for $x\notin S$.

In \emph{the fourth, final, part of the proof} we derive the bounds for the measurements. We do it by connecting the distance between the experimental measurements and the distance between the corresponding states. Indeed, we can write the following:
\beq
\norm{U(\tilde{M}^y_1)^{(T)}U^\dagger - M^y_1}\leq \norm{U(\tilde{\rho}^y_1)^{(T)}U^\dagger - \rho^y_1} + \norm{\tilde{\rho}^y_1-\tilde{M}^y_1},\quad \forall y,
\eeq
where we used the triangle inequality and the invariance of the norm $\norm{\tilde{\rho}^y_1-\tilde{M}^y_1}$ under unitary transformations. Remembering that $\norm{\tilde{\rho}^y_1-\tilde{M}^y_1}\leq \sqrt{\epsilon}$ (Lemma~\ref{lemma:overlaps}, Eq.~(\ref{app:eq_refined_norm_bound})) and $\norm{U(\tilde{\rho}^y_1)^{(T)}U^\dagger - \rho^y_1} = \td{U(\tilde{\rho}^y_1)^{(T)}U^\dagger}{\rho^y_1}$, we produce the bounds in Eqs.~(\ref{eq:rob_m_1},\ref{eq:rob_m_2}).
\end{proof}

Finally, we proceed to the proof of Theorem \ref{th:rob_quant_fid}, which gives a quantitative robustness analysis expressed in terms of average fidelity.

\begin{proof}[Proof of Theorem \ref{th:rob_quant_fid}]

Let us first state a useful identity between Frobenius distance and the fidelity, valid for an arbitrary pure states $\rho$ and a Hermitian operator $X$:
\begin{equation}\label{eq:fidFROB}
    \norm{X-\rho}^2_F=1+\tr(X^2)-2\tr(X\rho)\ .
\end{equation}

We follow exactly the same steps that were given in the proof of Theorem~\ref{th:rob_quant}. We assume that the reader is familiar with the notation introduced there.  First, in order to derive the bound in Eq.~\eqref{eq:rob_1_2} we repeat the reasoning preceding Eq.~\eqref{eq:ap_fid_u}. Analogously to the trace distance discussed there, we use the invariance of the fidelity under unitary transformations and transposition which gives us for all $x\in S=[k]$:
\begin{equation}
    \tr(\tilde{\tau}^x \tau^x  )= \tr(U (\tilde{\rho}^x_1)^{(T)} U^\dagger \rho^x_1  )\ ,
\end{equation}
where $U =V^\dagger\tilde{V}$. Using the above formula and standard algebra involving Pauli matrices we obtain:
\begin{equation}\label{eq:fisQUB}
\frac{1}{k}\sum_{x=1}^k \tr(U (\tilde{\rho}^x_1)^{(T)} U^\dagger \rho^x_1  ) = \frac{1}{2}+\frac{1}{2k}\sum_{x=1}^k  \mathbf{l}^x \cdot \mathbf{\tilde{l}}^x \ .
\end{equation}
On the other hand, we have the following identity:
\beq
\norm{\Delta L}^2_F = \sum_{x=1}^k |\mathbf{l}^x|^2 + \sum_{x=1}^k |\mathbf{\tilde{l}}^x|^2- 2 \sum_{x=1}^k  \mathbf{l}^x \cdot \mathbf{\tilde{l}}^x\ ,
\eeq
where $|\cdot|$ is a standard Euclidean norm in $\Rl^3$. Using the identity  $|\mathbf{\tilde{l}}^x|^2=2\tr((\tilde{\rho}^x_1)^2)-1$ and the fact that target states are pure (which implies $|\mathbf{l}^x|=1$) we obtain: 
\beq
\norm{\Delta L}^2_F = 2 \sum_{x=1}^k \tr((\tilde{\rho}^x_1)^2) - 2 \sum_{x=1}^k  \mathbf{l}^x \cdot \mathbf{\tilde{l}}^x\ .
\eeq
Inserting the above into Eq.~\eqref{eq:fisQUB} and using the fact that $\tr{(\tilde{\rho}^x_1)^2}\geq 1 - \frac{2\epsilon(1-2\epsilon)}{(1-\epsilon)^2} $ (this follows straightforwardly from Lemma \ref{lemma:norm}), we finally obtain:
\begin{equation}\label{eq:final_fid_bound}
 \frac{1}{k}\sum_{x=1}^k \tr(U (\tilde{\rho}^x_1)^{(T)} U^\dagger \rho^x_1  ) \geq 1- \frac{\epsilon(1-2\epsilon)}{(1-\epsilon)^2} -  \frac{\norm{\Delta L}^2_F}{4k}\ .
\end{equation}
We complete the proof of Eq.~\eqref{eq:rob_1_2} by again employing, exactly as before,  Theorem [Sun 1991] (Ref.~\cite{sun1991perturbation}) in order to upper-bound $\norm{\Delta L}_F$. 

Now, to derive Eq.~\eqref{eq:rob_2_2}, which gives the bounds for the fidelity of states for $x\in [k]$ and $a=2$, we can use the following inequality which  can be derived from Eq.~\eqref{eq:fidFROB}  and from the relation $\tr{(\tilde{\rho}^x_1)^2}\geq 1 - \frac{2\epsilon(1-2\epsilon)}{(1-\epsilon)^2}$:
\beq
\label{eq:tr_dist_fid_relation}
\tr(U (\tilde{\rho}^x_a)^{(T)} U^\dagger \rho^x_a  ) \geq 1 - \frac{\epsilon(1-2\epsilon)}{(1-\epsilon)^2}-(\td{U (\tilde{\rho}^x_a)^{(T)} U^\dagger}{\rho^x_a})^2.
\eeq
From Eqs.~(\ref{eq:ap_fid_bound_3},\ref{eq:ap_fid_a2}) it follows that: 
\beq
&&\frac{1}{k}\sum_{x=1}^k(\td{\tilde{\rho}^x_2)^{(T)} U^\dagger}{\rho^x_2})^2 \leq \frac{1}{k}\sum_{x=1}^k(\td{\tilde{\rho}^x_1)^{(T)} U^\dagger}{\rho^x_1}+2\sqrt{\epsilon})^2 \\
&&\leq \frac{\norm{\Delta L}^2_F}{4k} + \frac{2\norm{\Delta L}_F}{\sqrt{k}}\sqrt{\epsilon}+4\epsilon = \left(\frac{\norm{\Delta L}_F}{2\sqrt{k}}+2\sqrt{\epsilon}\right)^2, \nonumber
\eeq
which gives the desired bound form Eq.~\eqref{eq:rob_2_2}:
\beq\label{eq:final_fid_bound_2}
 \frac{1}{k}\sum_{x=1}^k \tr(U (\tilde{\rho}^x_2) ^{(T)} U^\dagger \rho^x_2  ) \geq 1- \frac{\epsilon(1-2\epsilon)}{(1-\epsilon)^2} -  \left(\frac{\norm{\Delta L}_F}{2\sqrt{k}}+2\sqrt{\epsilon}\right)^2\ .
\eeq
The proof of the remaining bounds in Eqs.~(\ref{eq:rob_3_2},\ref{eq:rob_m_1_2},\ref{eq:rob_m_2_2}) is straightforward and follows directly form  the formula in Eq.~\eqref{eq:fidFROB}. In particular, in order to prove Eq.~\eqref{eq:rob_3_2} we set $X=\tilde{\rho}^x_a$ and $\rho=\rho^x_a$ in Eq.~(\ref{eq:fidFROB}), use the inequality $\tr{(\tilde{\rho}^x_a)^2}\geq 1 - \frac{2\epsilon(1-2\epsilon)}{(1-\epsilon)^2}$ and the relation   $\norm{\sigma-\rho}^2_F = 2\td{\rho}{\sigma}^2$, with latter being true for arbitrary qubits states $\rho$ and $\sigma$. Derivation of Eq.~\eqref{eq:rob_m_1_2} is completely analogous to the one of Eq.~\eqref{eq:rob_2_2}. Finally, for Eq.~\eqref{eq:rob_m_2_2}, the reasoning is again analogous to Eq.~(\ref{eq:rob_3_2}), where, after setting  $X=\tilde{M}^y_1$ and $\rho=\rho^y_1$, it is additionally necessary to use a simple lower bound: $\tr((M^y_1)^2)\geq(1-\epsilon)^2$, which follows from the conditions of Theorem~\ref{th:rob_quant_fid}.
\end{proof}

\section{Alternative bounds from Procrustes}
\label{app:procrust}

\begin{lemma}[Robust self-testing for qubits from Procrustes]
\label{lemma:procrust}
Consider pure target qubit preparation states $\rho^x_a$ and target projective measurements $\Mb^y$,  where $a=1,2$ and $x,y\in [n]$.  Assume that $\rho^x_a = M^x_a$ for all $a,x$ and, furthermore, that experimental states $\tilde{\rho}^x_a$ and measurements $\Mbt^y$ act on Hilbert space of dimension at most $d$ and generate statistics $\tilde{p}(b|a,x,y)=\tr(\tilde{\rho}^x_a\tilde{M}^y_b)$ such that $|\tilde{p}(b|a,x,y)- \tr(\rho^x_a M^y_b)|\leq \epsilon$, for all  $a,b,x,y$. 

Let $\{\rho_i\}_{i=1}^m$ be a subset of $m$ considered states among $\rho^x_a$. Let $L$ be a matrix whose rows are the Bloch vectors of states $\rho_i$, $i\in [m]$, and let $k\in\lbrace{2,3\rbrace}$ be its rank ($m\geq k$). Assume, without loss of generality, that for $k=2$ the third component of the Bloch vectors is $0$. In that case, truncate $L$ to the first two columns. Let us define two auxiliary functions:
\beq
P_m(\epsilon,L) =  \begin{cases} \norm{L^\ddagger}F_m(\epsilon)+\min\left[\frac{\norm{L^\ddagger}F_m(\epsilon)}{\sqrt{1-\norm{L^\ddagger}^2F_m(\epsilon)}},\sqrt{\sqrt{k}F_m(\epsilon)}\right], & \text{if} \norm{L^\ddagger}\sqrt{F_m(\epsilon)} < 1 \\
\norm{L^\ddagger}F_m(\epsilon)+\sqrt{\sqrt{k}F_m(\epsilon)} & \text{otherwise,}
\end{cases}
\eeq
and
\beq
F_m(\epsilon) = \sqrt{4m(m-1)\epsilon\left(1+2\sqrt{\epsilon}+\frac{m+3}{m-1}\epsilon\right)},
\eeq
where $L^\ddagger = (LL^T)^{-1}L^T$. There exist a unitary matrix $U$ such that:
\beq
\frac{1}{m}\sum_{i=1}^m\tr(\rho_i U\tilde{\rho}_iU^\dagger) \geq 1-\frac{\epsilon(1-2\epsilon)}{(1-\epsilon)^2}-\frac{1}{4m}P^2_m(\epsilon,L).
\eeq

\end{lemma}

In this section we derive alternative bounds for the fidelity between preparation states that follows from the bounds on the so-called orthogonal Procrustes problem~\cite{hurley1962procrustes}. The problem itself can be formulated as follows. Given two sets of vectors $\mathbf{x}_1,\mathbf{x}_2,\dots,\mathbf{x}_m$ and $\mathbf{y}_1,\mathbf{y}_2,\dots,\mathbf{y}_m$ in $\Rl^d$ find an orthogonal transformation $O\in \mathrm{O}(d)$ in $\Rl^d$ that minimizes $\sum_{i=1}^m|\mathbf{x}_i-O\mathbf{y}_i|$. This problem has a clear relevance to our task. Indeed, if we take $\mathbf{x}_i$ to be the Bloch vectors of the target qubit preparation states and $\mathbf{y}_i$ the Bloch vectors of the experimental states, then minimization over $\mathrm{O}(3)$ is the same as the problem of finding a unitary transformation that connects those qubit states. In Ref.~\cite{arias2020perturbation} the bounds on Procrustes problem were derived. We give formulation of Theorem 1 from Ref.~\cite{arias2020perturbation} below, where we change the notation according to our problem.

\begin{theorem*}[Arias-Castro et.al. 2020 - A perturbation bound for Procrustes]
Given two tall matrices $L$ and $\tilde{L}$ of same size with $L$ having full rank, and set $\delta^2 = \norm{\tilde{L}\tilde{L}^T-LL^T}_F$. Then we have 
\beq
\label{app:eq_procr}
\min_{O\in\mathrm{O}(k)}\norm{L-O\tilde{L}}_F\leq \begin{cases} \norm{L^\ddagger}\delta^2+\min\left[\frac{\norm{L^\ddagger}\delta^2}{\sqrt{1-\norm{L^\ddagger}^2\delta^2}},k^\frac{1}{4}\delta\right], & \text{if} \norm{L^\ddagger}\delta < 1 \\
\norm{L^\ddagger}\delta^2+k^\frac{1}{4}\delta & \text{otherwise.} \end{cases}
\eeq
\end{theorem*}
\noindent In the formulation of the above theorem $L^\ddagger$ stands for the Moore–Penrose inverse, which can be defined as $L^\ddagger = (LL^T)^{-1}L^T$ for tall matrices of full rank.  

\begin{proof}[Proof of Lemma~\ref{lemma:procrust}]. Unitizing the results of the Theorem [Arias-Castro et.al. 2020] is rather straightforward. Let $L$ be a matrix which rows are the Bloch vectors of all target preparation states $\{\rho^x_a\}_{a,x}$. Let $\tilde{L}$ in turn be the matrix of Bloch vectors of the experimental states $\{\tilde{\rho}^x_a\}_{a,x}$. The matrices $LL^T$ and $\tilde{L}\tilde{L}^T$ are then, of course, the full Gram matrices $\Gamma$ and $\tilde{\Gamma}$. By ``full" we mean that now we do not select a subset $S$ of linearly independent vectors among the Bloch vectors of $\rho^x_a$. Instead, $\Gamma$ and $\tilde{\Gamma}$ are formed by all the considered states, which can still be a subset of the $2n$ states ($x\in [n]$, $a\in [2]$). We will be using $m$ to denote the number of the considered states, and a simple one-indexed set $\{\rho_i\}_i$ to denote the states themselves. Estimating $\delta^2$ from the formulation of the Theorem [Arias-Castro et.al. 2020] is a direct application of the bound on $\norm{\Delta\Gamma}_F$ from Lemma~\ref{lemma:norm}, where now instead of $k$ one should put the number $m$ of the considered preparation states.

As for the left-hand side of Eq.~(\ref{app:eq_procr}), we can write the following:
\beq
&&\norm{L-O\tilde{L}}_F^2 = \tr(LL^T)+\tr(\tilde{L}\tilde{L}^T) - 2\tr(L^TO\tilde{L}) \nonumber\\
&&\geq m+m-m\frac{4\epsilon(1-2\epsilon)}{(1-\epsilon)^2}-4\sum_{i=1}^m\tr(\rho_iU\tilde{\rho}_iU^\dagger)+2m
\eeq
where we use the following identity:
\beq
\tr(L^TO\tilde{L}) = \sum_{i=1}^m \mathbf{n}_i\cdot O\mathbf{\tilde{n}}_i = 2\sum_{i=1}^m\tr(\rho_iU\tilde{\rho}_iU^\dagger)-m,
\eeq
and $U$ is the unitary transformation in $\mathrm{SU}(2)$ corresponding to the orthogonal transformation $O$ of the Bloch vectors. The bound on the average fidelity between $m$ target preparation states and the corresponding experimental states is then simply:
\beq
\frac{1}{m}\sum_{i=1}^m\tr(\rho_i U\tilde{\rho}_iU^\dagger) \geq 1-\frac{\epsilon(1-2\epsilon)}{(1-\epsilon)^2}-\frac{1}{4m}P^2_m(\epsilon,L)
\eeq
where 
\beq
P_m(\epsilon,L) =  \begin{cases} \norm{L^\ddagger}F_m(\epsilon)+\min\left[\frac{\norm{L^\ddagger}F_m(\epsilon)}{\sqrt{1-\norm{L^\ddagger}^2F_m(\epsilon)}},\sqrt{\sqrt{k}F_m(\epsilon)}\right], & \text{if} \norm{L^\ddagger}\sqrt{F_m(\epsilon)} < 1 \\
\norm{L^\ddagger}F_m(\epsilon)+\sqrt{\sqrt{k}F_m(\epsilon)} & \text{otherwise,}
\end{cases}
\eeq
and
\beq
F_m(\epsilon) = \sqrt{4m(m-1)\epsilon\left(1+2\sqrt{\epsilon}+\frac{m+3}{m-1}\epsilon\right)}.
\eeq

\end{proof}

\section{Examples}
\label{app:ex}

Below we provide detailed derivations of the results presented in Table~\ref{tab:examples}.

The first example concerns $n=2,3$ MUBs in $d=2$. Since for MUBs $\tr(\rho^x_a\rho^{x'}_{a'}) = \frac{1}{2}$, for $x\neq x', \forall a,a'$, it follows that $\Gamma$ is an identity matrix in $\Rl^n$ ($n\in\{2,3\}$). Hence, in Theorem~\ref{th:rob_quant_fid} (see Appendix~\ref{app:rob}) we should take $\norm{\Gamma_S} = \norm{\Gamma^{-1}_S} = \norm{L_S} = 1$, and $\norm{\Gamma}_F = \sqrt{n}$. The resulting bound is the average between expressions in Eq.~(\ref{eq:rob_1}) and Eq.~(\ref{eq:rob_2}) with the function $E_{S,k}(\epsilon)$ being simply:
\beq
E_{S,k}(\epsilon) = \frac{1}{2\sqrt{2n}}\frac{F_k(\epsilon)}{\sqrt{1-O_k(\epsilon)}}.
\eeq
The leading linear term is given in Table~\ref{tab:examples} for both $n=2,3$.

The second example is a little less straightforward. From the condition $\tr(\rho^1_1\rho^2_1) = \frac{1+\alpha}{2}$, $\alpha \in (-1,1)$ we obtain that $\Gamma = \left(\begin{smallmatrix} 1 & \alpha \\ \alpha & 1 \end{smallmatrix} \right)$, and hence $\norm{\Gamma} = 1+|\alpha|$, $\norm{\Gamma^{-1}} = \frac{1}{1-|\alpha|}$, and $\norm{\Gamma}_F = \sqrt{2+2\alpha^2}$. The output $L$ of the Cholesky factorization is $L = \left(\begin{smallmatrix} 1 & 0 \\ \alpha & \sqrt{1-\alpha^2} \end{smallmatrix} \right)$, which leads to $\norm{L} = \sqrt{1+|a|}$. This also determines the minimum in Eq.~(\ref{eq:rob_E}) to be $\frac{\sqrt{1+|a|}}{\sqrt{2}}$. Plugging this values in Eq.~(\ref{eq:rob_E}) gives:
\beq
E_{\{1,2\},2}(\epsilon) = \frac{1}{4}\frac{\sqrt{1+|\alpha|}F_k(\epsilon)}{\sqrt{1-|\alpha|-O_k(\epsilon)}}.
\eeq
The final bound is again the average of the bounds in Eq.~(\ref{eq:rob_1}) and Eq.~(\ref{eq:rob_2}). The first order in $\epsilon$ for this bound is given in Table~\ref{tab:examples}. The applicability of the above bound is determined by the inequality $1-|\alpha|-O_2(\epsilon)\geq 0$. Importantly, the latter condition gives nonempty region $\epsilon\in [0,\epsilon_0)$ whenever $|\alpha|>0$.

The third example is a trine ensemble of states $(\rho^1_1,\rho^2_2,\rho^3_1)$, with $\rho^x_1=\frac{\Id}{2}+\frac{1}{2}\mathbf{n}_x\cdot\bm{\sigma}$, $x=1,2,3$, and where $\mathbf{n}_1 = (1,0,0)$, $\mathbf{n}_2 = \left(-\frac{1}{2},\frac{\sqrt{3}}{2},0\right)$, and $\mathbf{n}_3 = \left(-\frac{1}{2},-\frac{\sqrt{3}}{2},0\right)$. For this configuration of the preparation states the alternative robustness analysis via Procrustes (see Appendix~\ref{app:procrust}) gives better bounds. Given the vectors $\mathbf{n}_i$, $i=1,2,3$ we can directly compute that $\norm{L^\ddagger} = \sqrt{\frac{2}{3}}$. Inserting this value to Lemma~\ref{lemma:procrust} produces the results given in Table~\ref{tab:examples}. 

The fourth example is the tetrahedron, with $\mathbf{n}_1 = (0,0,1)$, $\mathbf{n}_2 = \left(\sqrt{\frac{8}{9}},0,-\frac{1}{3}\right)$, $\mathbf{n}_3 = \left(-\sqrt{\frac{2}{9}},\sqrt{\frac{2}{3}},-\frac{1}{3}\right)$, $\mathbf{n}_4 = \left(-\sqrt{\frac{2}{9}},-\sqrt{\frac{2}{3}},-\frac{1}{3}\right)$, and $\rho^x_1=\frac{\Id}{2}+\frac{1}{2}\mathbf{n}_x\cdot\bm{\sigma}$, $x=1,2,3,4$ as before. In this case, we also employ the bounds from Procrustes. For the above configuration of states we have that $\norm{L^\ddagger} = \frac{\sqrt{3}}{2}$, which leads to the results in Table~\ref{tab:examples}.

\section{Proofs of auxiliary results}\label{app:extra}

\begin{proof}[Proof of Lemma~\ref{lemma:overlaps}] First of all, we improve the bound on the norm of each of the experimental states. From $\norm{\tilde{M}^y_b}\geq 1-\epsilon$, for $d=2$ it follows immediately that the second (second largest) eigenvalue of each of the effects $\tilde{M}^y_b$ cannot exceed $\epsilon$. Hence, we can conclude that $\tr(\tilde{M}^y_b)\leq 1+\epsilon, \forall y,b$. 

Secondly, we can improve the bound on the norm of each of the experimental states. For that let us write the spectral decomposition of each experimental state and POVM effect:
\begin{align}
\tilde{\rho}^x_a &= \eta(\Id-\proj{\psi})+(1-\eta)\proj{\psi} = \eta\Id+(1-2\eta)\proj{\psi}\\
\tilde{M}^x_a &= \lambda_0\proj{\phi}+\lambda_1(\Id-\proj{\phi}),\nonumber
\end{align}
where we assume that $\lambda_0\geq \lambda_1$. We omitted the indices $x,a$ for $\eta$, $\lambda_0,\lambda_1$ and $\psi$,$\phi$ for simplicity. We can assume, without loss of generality, that $\eta\geq \frac{1}{2}$, and let us also assume for now that $|\braket{\phi}{\psi}|^2\leq \frac{1}{2}$. From the condition $\tr(\tilde{\rho}^x_a\tilde{M}^x_a)\geq 1-\epsilon$, it then follows that:
\beq
\label{eq:app_lemma_eta}
\eta(\lambda_0+\lambda_1)+(1-2\eta)(\lambda_0|\braket{\phi}{\psi}|^2+\lambda_1(1-|\braket{\phi}{\psi}|^2))\geq 1-\epsilon,
\eeq
from where we can obtain a lower bound on $\eta$, namely: 
\beq
\eta \geq \frac{1}{1-2|\braket{\phi}{\psi}|^2}\left(\frac{1-\epsilon-\lambda_1}{\lambda_0-\lambda_1}-\frac{1}{2}\right)+\frac{1}{2}.
\eeq
The expression on the right-hand side of the above inequality is maximal when $|\braket{\phi}{\psi}|^2 = 0$ and $\lambda_0 = 1$, $\lambda_1 = \epsilon$, which returns the bound $\norm{\tilde{\rho}^x_a} \geq \frac{1-2\epsilon}{1-\epsilon}$. 

Now let us return to our assumption $|\braket{\phi}{\psi}|^2\leq \frac{1}{2}$, for which the above bound is valid. We can upper-bound $\eta$ by $1$ in Eq.~(\ref{eq:app_lemma_eta}), which returns nontrivial upper bound on $|\braket{\phi}{\psi}|^2$ that happens to be $1-\frac{1-2\epsilon}{1-\epsilon}$. This function is below $\frac{1}{2}$ for $\epsilon\leq \frac{1}{3}$, i.e.,~for $\epsilon\leq \frac{1}{3}$ our newly-derived bound on $\norm{\tilde{\rho}^x_a}$ is valid. The region $\epsilon\in [0,\frac{1}{3}]$ is significantly larger than the resulting region in which our self-testing argument are valid, so this assumption does not affect our final results.

Let us now try to improve the bounds for the overlaps. In the proof of Theorem~\ref{th:overlaps} (see Eq.~(\ref{eq:app_proof_overlaps})) we have already established that:
\beq
|\tr(\tilde{\rho}^x_a\tilde{\rho}^{x'}_{a'})-\tr(\rho^x_a\rho^{x'}_{a'})| \leq  \epsilon + \norm{\tilde{\rho}^x_a-\tilde{M}^x_a},\quad \forall x\neq x',\forall a,a'.
\eeq
Let us now refine the bound on $\norm{\tilde{\rho}^x_a-\tilde{M}^x_a}$. For simplicity, we present this result below for a pair of operators $\rho = (1-\eta)\Id+(2\eta-1)\proj{\psi}$ and $M = \lambda_1\Id+(\lambda_0-\lambda_1)\proj{\phi}$, that satisfy the conditions $\tr(\rho M)\geq 1-\epsilon$, and $1-\epsilon\leq\lambda_0\leq 1$, $0\leq \lambda_1\leq \epsilon$. To compute the norm, we look for the eigenvector $\ket{\xi}$ of the operator $(\rho-M)$, i.e.,  $(\rho-M)\ket{\xi} = \Lambda\ket{\xi}$, corresponding to the maximal eigenvalue. We have the following quadratic equation for the eigenvalue $\Lambda$:
\beq
\label{eq:eig}
&&\hspace{-16pt}\Lambda^2-\Lambda(1-\lambda_0-\lambda_1)+\eta(1-\eta)-\lambda_1+\lambda_0\lambda_1-\eta(\lambda_0-\lambda_1)+(2\eta-1)(\lambda_0-\lambda_1)|\bk{\phi}{\psi}|^2=0.\nonumber\\&&
\eeq 
The sum of the roots of this equation is equal to $1-\lambda_0-\lambda_1$. Since we know that $|1-\lambda_0-\lambda_1|\leq \epsilon$, then either both roots are of the same sign, in which case the largest eigenvalue could only be $\epsilon$, or they are of the opposite sign. In the latter case, it is evident that the absolute values of both roots are maximal whenever the free term in Eq.~(\ref{eq:eig}) is minimal. We then can upper-bound the solutions to Eq.~(\ref{eq:eig}) by lower-bounding the free term using the condition $\tr(\rho M)\geq 1-\epsilon$. Indeed, from $\tr(\rho M)\geq 1-\epsilon$ we infer immediately that:
\beq
\tr(\rho M) = \lambda_0-\eta(\lambda_0-\lambda_1)+(2\eta-1)(\lambda_0-\lambda_1)|\bk{\phi}{\psi}|^2 \geq 1-\epsilon,
\eeq
and thus we reduce Eq.~(\ref{eq:eig}) to:
\beq
\label{eq:eig2}
\Lambda^2-\Lambda(1-\lambda_0-\lambda_1)+\eta(1-\eta)-\lambda_0 -\lambda_1+\lambda_0\lambda_1+1-\epsilon = 0.
\eeq
Using the same argument we can set $\eta=1$, because $\eta(1-\eta)\geq 0$. The positive root of Eq.~(\ref{eq:eig2}) is equal to:
\beq
\Lambda = \frac{1-\lambda_0-\lambda_1}{2}+\sqrt{\left(\frac{1-\lambda_0-\lambda_1}{2}\right)^2-(1-\lambda_0)(1-\lambda_1)+\epsilon}.
\eeq
It is easy to check that the above expression does not have any local maxima w.r.t. $\lambda_0$,$\lambda_1$ on the domain $1-\epsilon\leq\lambda_0\leq 1$, $0\leq\lambda_1\leq \epsilon$, whenever $\epsilon< \frac{1}{2}$, which we assume to be the case. Thus, we conclude that the maximal value of $\Lambda$ corresponds to the boundary of the region of $(\lambda_0,\lambda_1)$. By considering this boundary we find that this maximal value corresponds to the case of $\lambda_0=1$ and $\lambda_1=0$ which yields $\Lambda = \sqrt{\epsilon}$. From the above argument we finally conclude that: 
\beq
\label{app:eq_refined_norm_bound}
\norm{\tilde{\rho}^x_a-\tilde{M}^x_a}\leq \sqrt{\epsilon},\quad \forall a,
\eeq
which completes the proof for the state overlaps. From Eq.~(\ref{app:eq_refined_norm_bound}) and the fact that $\tr(\tilde{M}^y_b)\leq 1+\epsilon$, it is easy to obtain the improved bound on the overlaps between measurement effect.
\end{proof}

\begin{proof}[Proof of Lemma~\ref{lemma:norm}] Let us start by deriving the bound on the Frobenius norm $\norm{\Delta\Gamma}_F$:
\beq
\label{eq:app_frob_norm_der_1}
&&\norm{\Delta\Gamma}^2_F = \sum_{x=1}^{k}(|\Gamma_{x,x}-\tilde{\Gamma}_{x,x}|^2+\sum_{x'\neq x}|\Gamma_{x,x'}-\tilde{\Gamma}_{x,x'}|^2) \nonumber\\
&&= 4\sum_{x=1}^{k}\Big((1-\tr(\tilde{\rho}^x_1)^2)^2+\sum_{x'\neq x}|\tr(\rho^x_1\rho^{x'}_1)-\tr(\tilde{\rho}^x_1\tilde{\rho}^{x'}_1)|^2\Big).
\eeq
We have already established the bound on $|\tr(\rho^x_1\rho^{x'}_1)-\tr(\tilde{\rho}^x_1\tilde{\rho}^{x'}_1)|$ in Lemma~\ref{lemma:overlaps}. The bound on $(1-\tr(\tilde{\rho}^x_1)^2)^2$ can be obtained from the bound on the norm of each $\tilde{\rho}^x_1$. Namely, from the condition $\norm{\tilde{\rho}^x_a}\geq \frac{1-2\epsilon}{1-\epsilon}$ we can immediately conclude that:
\beq
1-\tr(\tilde{\rho}^x_1)^2 \leq \frac{2\epsilon(1-2\epsilon)}{(1-\epsilon)^2}.
\eeq
From here, it is easy to get to the final bound on $\norm{\Delta\Gamma}^2_F $, which reads:
\beq
\norm{\Delta\Gamma}^2_F \leq 4k(k-1)(\epsilon+\sqrt{\epsilon})^2+\frac{16k\epsilon^2(1-2\epsilon)^2}{(1-\epsilon)^4} \leq 4k(k-1)\epsilon\left(1+2\sqrt{\epsilon}+\frac{k+3}{k-1}\epsilon\right),
\eeq
where we made some approximations to simplify the result.

Now, let us derive the bound on $\norm{\Delta\Gamma}$. In principle, we know that $\norm{\Delta\Gamma}\leq \norm{\Delta\Gamma}_F$, but we can derive a better bound based on the fact that the diagonal entries of $\Delta\Gamma$ are much less than the off-diagonal entries.  

First of all, due to the triangle inequality, we can write $\norm{\Delta\Gamma}\leq \norm{\mathrm{diag}(\Delta\Gamma)}+\norm{\mathrm{offdiag}(\Delta\Gamma)}$, where we split $\Delta\Gamma$ on the diagonal and off-diagonal parts. The first term, $\norm{\mathrm{diag}(\Delta\Gamma)}$, can be easily bounded as follows:
\beq
\norm{\mathrm{diag}(\Delta\Gamma)} = 2\max_x(1-\tr(\tilde{\rho}^x_1)^2) \leq \frac{4\epsilon(1-2\epsilon)}{(1-\epsilon)^2}\leq 4\epsilon.
\eeq
As for the off-diagonal part, we give the proof for two cases $k=2$ and $k=3$ separately. For $k=2$, $\norm{\mathrm{offdiag}(\Delta\Gamma)} = 2|\tr(\rho^1_1\rho^{2}_1)-\tr(\tilde{\rho}^1_1\tilde{\rho}^{2}_1)|\leq 2\sqrt{\epsilon}+2\epsilon$, where we used the results of Lemma~\ref{lemma:overlaps}. As for $k=3$, we will need some intermediate result, namely the following relation:
\beq
\label{eq:frob_oper_norm}
\norm{A}\leq \sqrt{\frac{k-1}{k}}\norm{A}_F,
\eeq
where $k$ is the size of the matrix $A$ with $\tr(A)=0$. We give a proof of this below.

Let us assume that $\{\lambda_i\}_{i=1}^k$ are the eigenvalues of matrix $A$, hence we know that $\sum_{i=1}^k\lambda_i = 0.$ Let $\lambda_1$ be the largest eigenvalue, i.e.,~the operator norm of $A$, if $\lambda_1\geq 0$. If we wish to maximize the Frobenius norm of $A$ for fixed $\lambda_1$, the following lower-bound has to be satisfied:
\beq
\norm{A}_F^2 = \lambda_1^2+\sum_{i=2}^k\lambda_i^2\geq \lambda_1^2+\frac{1}{k-1}\left(\sum_{i=2}^k|\lambda_i|\right)^2\geq \lambda_1^2+\frac{1}{k-1}\left(\sum_{i=2}^k\lambda_i\right)^2 = \lambda_1^2\frac{k}{k-1},
\eeq
which proves the bound in Eq.~(\ref{eq:frob_oper_norm}). Using the above result, we obtain the following bound:
\beq
&&\hspace{-18pt}\norm{\mathrm{offdiag}(\Delta\Gamma)}\leq \sqrt{\frac{2}{3}}\norm{\mathrm{offdiag}(\Delta\Gamma)}_F = 2\sqrt{\frac{2}{3}}\sqrt{\sum_{x\neq x'}|\tr(\rho^x_1\rho^{x'}_1)-\tr(\tilde{\rho}^x_1\tilde{\rho}^{x'}_1)|^2}\leq 4(\sqrt{\epsilon}+\epsilon),\nonumber\\
&&
\eeq
which completes out proof.
\end{proof}

\bibliographystyle{unsrt}
\bibliography{self_test_prep_meas}

\end{document}